\documentclass[submission,copyright,creativecommons]{eptcs}
\providecommand{\event}{QPL 2020} % Name of the event you are submitting to
 \usepackage{amssymb,nicefrac}
\usepackage{amsfonts}
\usepackage{graphicx}
\usepackage[fleqn]{amsmath}
\usepackage[varg]{txfonts}
\usepackage{stmaryrd}%for symbol of standard interpretation
\usepackage{framed}
\usepackage{color}
\usepackage{tikzfig}
\usepackage{placeins}
\usepackage[normalem]{ulem} % remove underline from references tiltes with option [normalem]
\usepackage{mathtools} 
\DeclarePairedDelimiter\bra{\langle}{\rvert}
\DeclarePairedDelimiter\ket{\lvert}{\rangle}
\newenvironment{proof}{\textbf{Proof:}}{\hfill$\Box$\newline}
\tikzstyle{env}=[copoint,regular polygon rotate=0,minimum width=0.2cm, fill=black]

\tikzstyle{every picture}=[baseline=-0.25em]
\tikzstyle{dotpic}=[scale=0.5]
\tikzstyle{diredges}=[every to/.style={diredge}]
\tikzstyle{dot graph}=[shorten <=-0.1mm,shorten >=-0.1mm,scale=0.6]
\tikzstyle{plot point}=[circle,fill=black,minimum width=2mm,inner sep=0]

\tikzstyle{braceedge}=[decorate,decoration={brace,amplitude=2mm,raise=-1mm}]
\tikzstyle{small braceedge}=[decorate,decoration={brace,amplitude=1mm,raise=-1mm}]
\tikzstyle{left hook arrow}=[left hook-latex]
\tikzstyle{right hook arrow}=[right hook-latex]

\tikzstyle{dtriangle}=[fill=yellow,draw=black,shape=isosceles triangle,shape border rotate=-90,isosceles triangle stretches=true,inner sep=0.8pt,minimum width=0.25cm,minimum height=2mm]
\tikzstyle{vtriang}=[fill=yellow,draw=black,shape=isosceles triangle,shape border rotate=180,isosceles triangle stretches=true,inner sep=0.8pt,minimum width=0.25cm,minimum height=2mm]
\tikzstyle{vrt}=[fill=yellow,draw=black,shape=isosceles triangle,shape border rotate=0,isosceles triangle stretches=true,inner sep=0.8pt,minimum width=0.25cm,minimum height=2mm]
\tikzstyle{H box}=[rectangle,fill=yellow,draw=black,xscale=0.8,yscale=0.8, inner sep=0.6pt]
\tikzstyle{gbox}=[rectangle,fill=green,draw=black,xscale=1.0,yscale=1.0, inner sep=1.pt]
\tikzstyle{rbox}=[rectangle,fill=red,draw=black,xscale=1.0,yscale=1.0, inner sep=1.pt]
\tikzstyle{square box}=[rectangle,fill=white,draw=black,xscale=1.0,yscale=1.0, inner sep=1.pt]
\tikzstyle{whn}=[dot,fill=white,minimum width=0.25cm,inner sep=0pt]
\tikzstyle{zhbx}=[rectangle,fill=white,draw=black,xscale=1.0,yscale=1.0, inner sep=1.6pt]
\tikzstyle{gray dot}=[dot,fill=gray!40!white,minimum width=0.25cm,inner sep=0pt]
\tikzstyle{triangle}=[fill=yellow,draw=black,shape=isosceles triangle,shape border rotate=90,isosceles triangle stretches=true,inner sep=0.8pt,minimum width=0.25cm,minimum height=2mm]
\tikzstyle{bn}=[circle,fill=black,draw=black,scale=.4]
\tikzstyle{wn}=[circle,fill=white,draw=black,scale=.6]
\tikzstyle{dn}=[circle,fill=none,draw=gray]

\tikzstyle{Z dot}=[inner sep=0mm, minimum size=2mm, shape=circle, draw=black, fill={rgb,255: red,221; green,255; blue,221}]
\tikzstyle{Z phase dot}=[minimum size=5mm, font={\footnotesize\boldmath}, shape=rectangle, rounded corners=2mm, inner sep=0.2mm, outer sep=-2mm, scale=0.8, draw=black, fill={rgb,255: red,221; green,255; blue,221}]
\tikzstyle{X dot}=[Z dot, shape=circle, draw=black, fill={rgb,255: red,255; green,136; blue,136}]
\tikzstyle{X phase dot}=[Z phase dot, fill={rgb,255: red,255; green,136; blue,136}, font={\footnotesize\boldmath}]
\tikzstyle{hadamard edge}=[-, dashed, dash pattern=on 2pt off 0.5pt, thick, draw={rgb,255: red,68; green,136; blue,255}]
\tikzstyle{black dot}=[inner sep=0.7mm,minimum width=0pt,minimum height=0pt,fill=black,draw=black,shape=circle]

\tikzstyle{dot}=[black dot]
\tikzstyle{smalldot}=[inner sep=0.4mm,minimum width=0pt,minimum height=0pt,fill=black,draw=black,shape=circle]%NEW
\tikzstyle{white dot}=[dot,fill=white]
\tikzstyle{antipode}=[white dot,inner sep=0.3mm,font=\footnotesize]
\tikzstyle{smallwhitedot}=[smalldot,fill=white]%NEW
\tikzstyle{alt white dot}=[white dot,label={[xshift=3.07mm,yshift=-0.05mm,font=\footnotesize]left:$*$}]
\tikzstyle{smallgraydot}=[smalldot,fill=gray!40!white]%NEW
\tikzstyle{box vertex}=[draw=black,rectangle]
\tikzstyle{small box}=[box vertex,fill=white]%% added rwd]
\tikzstyle{whitebg}=[fill=white,inner sep=2pt]
\tikzstyle{graph state vertex}=[sg vertex,fill=black]

\tikzstyle{wide copoint}=[fill=white,draw=black,shape=isosceles triangle,shape border rotate=90,isosceles triangle stretches=true,inner sep=1pt,minimum width=1.5cm,minimum height=5mm]
\tikzstyle{wide point}=[fill=white,draw=black,shape=isosceles triangle,shape border rotate=-90,isosceles triangle stretches=true,inner sep=1pt,minimum width=1.5cm,minimum height=4mm]
\tikzstyle{very wide copoint}=[fill=white,draw=black,shape=isosceles triangle,shape border rotate=-90,isosceles triangle stretches=true,inner sep=1pt,minimum width=2.5cm,minimum height=4mm]
\tikzstyle{very wide empty copoint}=[draw=black,shape=isosceles triangle,shape border rotate=-90,isosceles triangle stretches=true,inner sep=1pt,minimum width=2.5cm,minimum height=4mm]
\tikzstyle{symm}=[ultra thick,shorten <=-1mm,shorten >=-1mm]

\tikzstyle{square gray box}=[rectangle,fill=gray!30,draw=black,minimum height=6mm,minimum width=6mm]
\tikzstyle{copoint}=[regular polygon,regular polygon sides=3,draw=black,scale=0.75,inner sep=-0.5pt,minimum width=7mm,fill=white]
\tikzstyle{point}=[regular polygon,regular polygon sides=3,draw=black,scale=0.75,inner sep=-0.5pt,minimum width=7mm,fill=white,regular polygon rotate=180]
\tikzstyle{gray point}=[point,fill=gray!40!white]
\tikzstyle{gray copoint}=[copoint,fill=gray!40!white]

\newcommand{\edgearrow}{{\arrow[black]{>}}}
\newcommand{\edgetick}{{\arrow[black,scale=0.7,very thick]{|}}}
\tikzstyle{diredge}=[->]
\tikzstyle{rdiredge}=[<-]
\tikzstyle{medium diredge}=[->]

\tikzstyle{short diredge}=[->]
\tikzstyle{halfedge}=[-)]
\tikzstyle{other halfedge}=[(-]
\tikzstyle{freeedge}=[(-)]
\tikzstyle{white edge}=[line width=5pt,white]
\tikzstyle{tick}=[postaction=decorate,decoration={markings, mark=at position 0.5 with \edgetick}]
\tikzstyle{small map edge}=[|-latex, gray!60!blue, shorten <=0.9mm, shorten >=0.5mm]
\tikzstyle{thick dashed edge}=[very thick,dashed,gray!40]
\tikzstyle{map edge}=[|-latex,very thick, gray!40, shorten <=1mm, shorten >=0.5mm]
\tikzstyle{tickedge}=[postaction=decorate,
  decoration={markings, mark=at position 0.5 with \edgetick}]
\tikzstyle{dirtickedge}=[postaction=decorate,
  decoration={markings, mark=at position 0.5 with \edgetick},
  decoration={markings, mark=at position 0.85 with \edgearrow}]
\tikzstyle{dirdoubletickedge}=[postaction=decorate,
  decoration={markings, mark=at position 0.4 with \edgetick},
  decoration={markings, mark=at position 0.6 with \edgetick},
  decoration={markings, mark=at position 0.85 with \edgearrow}]

\makeatletter
\newcommand{\boxshape}[3]{%
\pgfdeclareshape{#1}{
\inheritsavedanchors[from=rectangle] % this is nearly a rectangle
\inheritanchorborder[from=rectangle]
\inheritanchor[from=rectangle]{center}
\inheritanchor[from=rectangle]{north}
\inheritanchor[from=rectangle]{south}
\inheritanchor[from=rectangle]{west}
\inheritanchor[from=rectangle]{east}
\backgroundpath{% this is new
\southwest \pgf@xa=\pgf@x \pgf@ya=\pgf@y
\northeast \pgf@xb=\pgf@x \pgf@yb=\pgf@y

\@tempdima=#2
\@tempdimb=#3

\pgfpathmoveto{\pgfpoint{\pgf@xa - 5pt + \@tempdima}{\pgf@ya}}
\pgfpathlineto{\pgfpoint{\pgf@xa - 5pt - \@tempdima}{\pgf@yb}}
\pgfpathlineto{\pgfpoint{\pgf@xb + 5pt + \@tempdimb}{\pgf@yb}}
\pgfpathlineto{\pgfpoint{\pgf@xb + 5pt - \@tempdimb}{\pgf@ya}}
\pgfpathlineto{\pgfpoint{\pgf@xa - 5pt + \@tempdima}{\pgf@ya}}
\pgfpathclose
}
}}

\boxshape{NEbox}{0pt}{8pt}
\boxshape{SEbox}{0pt}{-8pt}
\boxshape{NWbox}{8pt}{0pt}
\boxshape{SWbox}{-8pt}{0pt}
\makeatother

\tikzstyle{map}=[draw,shape=NEbox,inner sep=7pt]
\tikzstyle{mapdag}=[draw,shape=SEbox,inner sep=7pt]
\tikzstyle{maptrans}=[draw,shape=SWbox,inner sep=7pt]
\tikzstyle{mapconj}=[draw,shape=NWbox,inner sep=7pt]

\tikzstyle{probs}=[shape=semicircle,fill=gray!40!white,draw=black,shape border rotate=180,minimum width=1.2cm]

\tikzstyle{arrs}=[-latex,font=\small,auto]
\tikzstyle{arrow plain}=[arrs]
\tikzstyle{arrow dashed}=[dashed,arrs]
\tikzstyle{arrow bold}=[very thick,arrs]
\tikzstyle{arrow hide}=[draw=white!0,-]
\tikzstyle{arrow reverse}=[latex-]
\tikzstyle{cdnode}=[]

\tikzstyle{gn}=[dot,fill=green,minimum width=0.25cm,inner sep=0pt]
\tikzstyle{rn}=[dot,fill=red,inner sep=0pt,minimum width=0.25cm]
\tikzstyle{rc}=[dot,thick,fill=white,draw = red,minimum width=0.3cm,inner sep=0pt]
\tikzstyle{gc}=[dot,thick,fill=white,draw= green,inner sep=0pt,minimum width=0.3cm]
\tikzstyle{bc}=[dot,thick,fill=white,draw= blue,minimum width=0.3cm]

\tikzstyle{label}=[circle,fill=white,minimum width=0.3cm]

\tikzstyle{clocklabel}=[dot,fill=yellow,draw=black,font=\tiny,inner sep=0.75pt]

\tikzstyle{rsn}=[circle split,draw,fill=red,font=\tiny,inner sep=0.75pt]
\tikzstyle{gsn}=[circle split,draw,fill=green,font=\tiny,inner sep=0.75pt]
\tikzstyle{bsn}=[circle split,draw,fill=blue,font=\tiny,inner sep=0.75pt]

\tikzstyle{rsc}=[circle split,thick,draw= red,draw,fill=white,font=\tiny,inner sep=0.75pt]
\tikzstyle{gsc}=[circle split,thick,draw= green,draw,fill=white,font=\tiny,inner sep=0.75pt]
\tikzstyle{bsc}=[circle split,thick,draw= blue,draw,fill=white,font=\tiny,inner sep=0.75pt]

\tikzstyle{cnot}=[fill=white,shape=circle,inner sep=-1.4pt]
\tikzstyle{wire label}=[font=\tiny, auto]

\tikzstyle{cdiag}=[matrix of math nodes, row sep=3em, column sep=3em, text height=1.5ex, text depth=0.25ex,inner sep=0.5em]
\tikzstyle{arrow above}=[transform canvas={yshift=0.5ex}]
\tikzstyle{arrow below}=[transform canvas={yshift=-0.5ex}]

\newtheorem{Th}{Theorem}[section]
\newtheorem{theorem}[Th]{Theorem}
\newtheorem{proposition}[Th]{Proposition} 
\newtheorem{lemma}[Th]{Lemma}
\newtheorem{corollary}[Th]{Corollary}
\newtheorem{definition}[Th]{Definition} 

\newtheorem{remark}[Th]{Remark}

\makeatletter
\newcommand{\vast}{\bBigg@{6.5}}
\makeatother

\title{An Algebraic Axiomatisation of ZX-calculus}
\author{Quanlong Wang
\institute{Cambridge Quantum Computing Ltd.}
\email{harny.wang@cambridgequantum.com}
}
\def\titlerunning{An Algebraic Axiomatisation of ZX-calculus}
\def\authorrunning{Q. Wang}
\begin{document}
\maketitle
%\date{}\maketitle
\begin{abstract}
%In this paper we give an  algebraic axiomatisation of ZX-calculus in the sense that there are only ring operations involved for phases, without any need of  trigonometry functions such as Sin and Cos, in contrast to previous universally complete axiomatisations of ZX-calculus. With this ZX-calculus, we have an easy translation from ZH-calculus to ZX-calculus  and vice versa.  We also derive all the rules of ZH-calculus within ZX-calculus after translation of diagrams.
% with the help of two additional quantum Boolean rules.
ZX-calculus is a graphical language for quantum computing which is complete in the sense that calculation in matrices can be done in a purely diagrammatic way. However, all previous universally complete axiomatisations of ZX-calculus have included at least one rule involving trigonometric functions such as $\sin$ and $\cos$ which makes it difficult for application purpose. In this paper we give an algebraic complete axiomatisation of ZX-calculus instead such that there are only ring operations involved for phases. With this algebraic axiomatisation of ZX-calculus, we are able to establish for the first time a simple translation of diagrams from another graphical language called ZH-calculus and to derive all the ZX-translated rules of ZH-calculus. As a consequence,  we have a great benefit that all techniques obtained in ZH-calculus can be transplanted to ZX-calculus, which can't  be obtained by just using the completeness of ZX-calculus.

\end{abstract}

% ===============================================================

%\iffalse
\section{Introduction}
%\textcolor{red}{More motivations on why the algebraic axiomatisation is important (the rules are non trivial, it could have good translation from zh generators and rules. we already used this algebraic axiomatisation in deriving P rule and spider nest identities, both are very important in applications, cite 2quibit, completeness and t-count papers.) and why concrete translation from ZH to ZX is important (first the translation is not trivial, since is not known before. without this translation, its impossible to use the completeness of zx to say that all the zh rules translated in zx can be derived, also derive them directly is very helpful in applications. ZH already has applications, e.g., paper in this qpl. so any technique can be transplanted to zx, otherwise it's impossible. and this benefit is not only from the translation, but also from the  generators and rules: without this generators, translation won't be so simple,; without the rules, the ZH rules would be very hard to derived in ZX.) }
%\textcolor{red}{revise this paper according to QPL slides}

The ZX-calculus was introduced by Coecke and Duncan \cite{CoeckeDuncan} as a graphical language for quantum computing, based on the framework of compact closed categories. The core part of ZX-calculus is a pair of spiders (complementary observables) with strong complementarity \cite{CDKW}. The ZX-calculus can also be seen as a form of PROP \cite{maclane1965}, thus it is usually presented by generators and rewriting rules.

There are three important properties of ZX-calculus: soundness, universality and completeness. Soundness means all the ZX rewriting rules hold when interpreted by matrices. Universality means each matrix (linear map between finite dimensional Hilbert Spaces)can be represented by a ZX diagram. Finally, completeness means each diagrammatic equality can be derived from ZX rules if their corresponding matrix equality holds in  finite dimensional Hilbert Spaces. The soundness and universality of  ZX-calculus have been proved in \cite{CoeckeDuncan}. The universal completeness  of ZX-calculus (which means ZX-calculus is complete for the full pure qubit quantum mechanics instead of any part of it) was first given in \cite{ngwang} and then incorporated in \cite{amarngwanglics}. The feature of this complete axiomatisation is that it has two new generators: the $\lambda$ box and the triangle symbol (which first appeared in \cite{jpvcltlics} as a short notation for some diagram composed of mere green and red nodes).  Based on some results in  \cite{ngwang},  there came another 
 universal complete axiomatisation of ZX-calculus  \cite{jpvbeyondlics}  with only traditional generators as given in  \cite{CoeckeDuncan}.  Thereafter, two more  universal complete axiomatisations of ZX-calculus  were presented  \cite{jpvnormfmlics},  \cite{Renaudprulelics}. All of these universal complete axiomatisations of ZX-calculus have some non-algebraic rule involved with trigonometry functions such as $\sin$ or $\cos$.
 For example, the following so-called (P) rule  \cite{coeckewang} is deployed in \cite{Renaudprulelics} (with scalars added) as a key rule for universal completeness.
 \begin{equation}\label{zxztoxzxcreq}
	\beginpgfgraphicnamed{TikZit/zxztoxzx}
	\InputIfFileExists{TikZit/zxztoxzx.tikz}{}{\input{./figures/TikZit/zxztoxzx.tikz}}%
	\endpgfgraphicnamed
\qquad\mbox{with}\quad
\left\{
\begin{array}{l}
\alpha_2=\arg z+\arg z_1\\
\beta_2=2\arg (|\frac{z}{z_1}|+i)\\
\gamma_2=\arg z-\arg z_1
\end{array}
\right.
\end{equation}
where:
\[
\begin{array}{l}
% \alpha_1, \beta_1, \gamma_1 \in (0, ~2\pi)\\
z=\cos\frac{\beta_1}{2}\cos\frac{\alpha_1+\gamma_1}{2}+i\sin\frac{\beta_1}{2}\cos\frac{\alpha_1-\gamma_1}{2}\\
z_1=\cos\frac{\beta_1}{2}\sin\frac{\alpha_1+\gamma_1}{2}-i\sin\frac{\beta_1}{2}\sin\frac{\alpha_1-\gamma_1}{2}
\end{array}
\]
 One could imagine that it would be very hard to use such a rule directly if there are trigonometry functions involved. For example, in (\ref{zxztoxzxcreq}), if we take $\alpha_1=\frac{\pi}{4},~  \beta_1=-\frac{\pi}{4}, \gamma_1=\frac{\pi}{2}, $ then by  tedious calculation one can get that $\alpha_2=arctan(-\sqrt{2}),~  \beta_2=-\frac{\pi}{3}, \gamma_2=arctan(\frac{-1}{\sqrt{2}});$ or $\alpha_2=\pi-arctan(\sqrt{2}),~  \beta_2=\frac{\pi}{3}, \gamma_2=\pi-arctan(\frac{1}{\sqrt{2}}).$ This means simple angles can be turned into complicated angles, thus not easy to deal with in applications.  

In this paper we overcome this drawback by giving a new complete axiomatisation of ZX-calculus with purely algebraic rules, in the sense that there are only ring operations involved for phases. One of the features of this axiomatisation in comparison to the previous ones is that we deploy a generator which is first introduced in \cite{ngwang2} diagrammatically represented by green box with parameters ranging in any complex numbers.   The usefulness of the new generator has been fully shown in  \cite{coeckewang} by its power in  deriving  the (P) rule. We obtain the completeness by deriving all the rules in \cite{amarngwanglics} from this new set of algebraic rules. One significant application of these  algebraic rules is the derivation of the so-called spider nest identities \cite{bobanthonywang}, which are key to the T-count reduction of quantum circuits \cite{nielbianwangtqc, nielbianwang}. The axiomatisation in  \cite{amarngwanglics} also has a triangle diagram as a generator as is the same case for the algebraic axiomatisation, whereas the former axiomatisation has a non-algebraic rule and its triangle-involved rules are more complicated then that of the latter. We point out that both the  green box and the triangle are not really external to the original generators usually described as green and red spiders \cite{CoeckeDuncan},  they  actually can be expressed in terms of those original generators, though in a complicated form  \cite{ngwang, amarngwanglics}.

Furthermore, for another graphical language for quantum computing called  ZH-calculus  \cite{miriamaleks},  which has applications in various fields like tensor network \cite{nakons2020} and hypergraph states  \cite{lwaleks2020},  we give a simple translation from the full ZH-calculus to ZX-calculus for the first time. Via this translation and the algebraic rules, we are able to derive all the ZX-translated ZH rules from the algebraic ZX rules. Although in principle we know this can be done because of the completeness of ZX-calculus,  that doesn't help us to gain a great benefit that all techniques obtained in ZH-calculus can be transplanted to ZX-calculus. Only by the detailed derivations  using the algebraic rules could we have such bonus for  ZX-calculus.

\section{Algebraic axiomatisation of ZX-calculus}
The ZX-calculus is based on a compact closed PROP \cite{maclane1965}, which  is a strict symmetric monoidal category whose objects are generated  by one object, with a compact structure  \cite{Coeckebk} as well. Each PROP can be described as a presentation in terms of  generators and relations \cite{BaezCR}.

%In linear algebra, there are three types of elementary row operations (respectively, column operations): row-switching transformations, row-multiplying transformations, row-addition transformations. Elementary row operations can be represented by left multiplication of elementary matrices. 

% {\bf General phase spider}

%\section{Algebraic rules of ZX-calculus}

First we give the generators of ZX-calculus in the following table. Note that all the diagrams in this paper should be read from top to bottom. 
%---------------------------------------------------------------------------------
\begin{table}[!h]
\begin{center} % \tikzfig{TikZit/emptysquare-small}
\begin{tabular}{|r@{~}r@{~}c@{~}c|r@{~}r@{~}c@{~}c|}
\hline
%$R_Z^{(n,m)}$&$:$&$n\to m$ & \tikzfig{TikZit/generator_spider2} & $A$&$:$&$ 1\to 1$& \tikzfig{TikZit/alphagate}\\
%$R_{Z,a}^{(n,m)}$&$:$&$n\to m$ & \tikzfig{TikZit/generalgreenspider}  & $R_{X,a}^{(n,m)}$&$:$&$n\to m$ & \tikzfig{TikZit/generalredspider}\\\hline
%\multicolumn{4}{lc}{ \tikzfig{TikZit/generator_spider2} }\\
%\hline
$R_{Z,a}^{(n,m)}$&$:$&$n\to m$ & %
	\beginpgfgraphicnamed{TikZit/generalgreenspider}
	\InputIfFileExists{TikZit/generalgreenspider.tikz}{}{\input{./figures/TikZit/generalgreenspider.tikz}}%
	\endpgfgraphicnamed
  & $H$&$:$&$1\to 1$ &%
	\beginpgfgraphicnamed{TikZit/HadaDecomSingleslt}
	\InputIfFileExists{TikZit/HadaDecomSingleslt.tikz}{}{\input{./figures/TikZit/HadaDecomSingleslt.tikz}}%
	\endpgfgraphicnamed
 \\\hline
 $\mathbb I$&$:$&$1\to 1$&%
	\beginpgfgraphicnamed{TikZit/Id}
	\InputIfFileExists{TikZit/Id.tikz}{}{\input{./figures/TikZit/Id.tikz}}%
	\endpgfgraphicnamed
 
%& $s$&$:$&$0\to 0$ &\tikzfig{scalars//halfscalar}  &
 &  $\sigma$&$:$&$ 2\to 2$& %
	\beginpgfgraphicnamed{TikZit/swap}
	\InputIfFileExists{TikZit/swap.tikz}{}{\input{./figures/TikZit/swap.tikz}}%
	\endpgfgraphicnamed
\\\hline
%   $\mathbb I$&$:$&$1\to 1$&\tikzfig{TikZit/Id} & $e $&$:$&$0 \to 0$& \tikzfig{TikZit/emptysquare}\\\hline
   $C_a$&$:$&$ 0\to 2$& %
	\beginpgfgraphicnamed{TikZit/cap}
	\InputIfFileExists{TikZit/cap.tikz}{}{\input{./figures/TikZit/cap.tikz}}%
	\endpgfgraphicnamed
 &$ C_u$&$:$&$ 2\to 0$&%
	\beginpgfgraphicnamed{TikZit/cup}
	\InputIfFileExists{TikZit/cup.tikz}{}{\input{./figures/TikZit/cup.tikz}}%
	\endpgfgraphicnamed
 \\\hline
  $T$&$:$&$1\to 1$&%
	\beginpgfgraphicnamed{TikZit/triangle}
	\InputIfFileExists{TikZit/triangle.tikz}{}{\input{./figures/TikZit/triangle.tikz}}%
	\endpgfgraphicnamed
  & $T^{-1}$&$:$&$1\to 1$&%
	\beginpgfgraphicnamed{TikZit/triangleinv}
	\InputIfFileExists{TikZit/triangleinv.tikz}{}{\input{./figures/TikZit/triangleinv.tikz}}%
	\endpgfgraphicnamed
 \\\hline
\end{tabular}\caption{Generators of ZX-calculus,where $m,n\in \mathbb N$, $ a  \in \mathbb C$.} \label{qbzxgenerator}
\end{center}
\end{table}
\FloatBarrier

 \begin{remark}
The generator $R_{Z,a}^{(n,m)}$ is first introduced in \cite{ngwang2}. It seems that the generators $R_{Z,a}^{(n,m)}$,  $T$ and $T^{-1}$ are totally external to the original generators usually described as green and red spiders \cite{CoeckeDuncan}, but they can actually be expressed in terms of those original generators, though in a complicated form  \cite{ngwang, amarngwanglics}. Here we just show that for $T$ and $T^{-1}$ as follows:
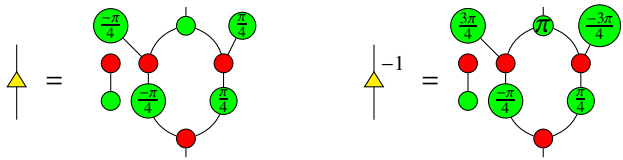
\[  %
	\beginpgfgraphicnamed{TikZit/triangleinvnspider}
	\InputIfFileExists{TikZit/triangleinvnspider.tikz}{}{\input{./figures/TikZit/triangleinvnspider.tikz}}%
	\endpgfgraphicnamed
\]
 \end{remark}

Also we define some diagrams as follows:

 \begin{equation*}\label{downtriangledef}
	\beginpgfgraphicnamed{TikZit/generalredspiderdefini}
	\InputIfFileExists{TikZit/generalredspiderdefini.tikz}{}{\input{./figures/TikZit/generalredspiderdefini.tikz}}%
	\endpgfgraphicnamed
 (H) \quad\quad\quad\quad %
	\beginpgfgraphicnamed{TikZit/downtriangle}
	\InputIfFileExists{TikZit/downtriangle.tikz}{}{\input{./figures/TikZit/downtriangle.tikz}}%
	\endpgfgraphicnamed

\end{equation*} 
For simplicity, we make the following conventions: 
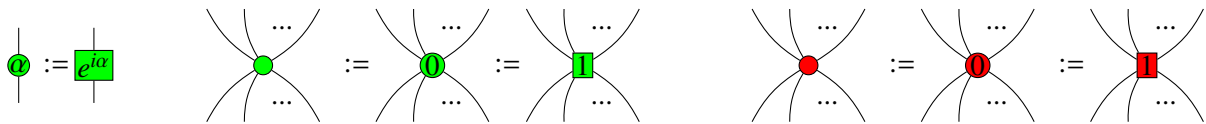
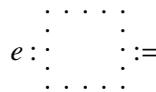
\[
	\beginpgfgraphicnamed{TikZit/spider0denote}
	\InputIfFileExists{TikZit/spider0denote.tikz}{}{\input{./figures/TikZit/spider0denote.tikz}}%
	\endpgfgraphicnamed
 
\]
and $$e: %
	\beginpgfgraphicnamed{TikZit/emptysquare}
	\InputIfFileExists{TikZit/emptysquare.tikz}{}{\input{./figures/TikZit/emptysquare.tikz}}%
	\endpgfgraphicnamed
:=$$
which means  $e$ represents an empty diagram.

There is a standard interpretation $\left\llbracket \cdot \right\rrbracket$ for the ZX diagrams:
\[
\left\llbracket %
	\beginpgfgraphicnamed{TikZit/generalgreenspider}
	\InputIfFileExists{TikZit/generalgreenspider.tikz}{}{\input{./figures/TikZit/generalgreenspider.tikz}}%
	\endpgfgraphicnamed
 \right\rrbracket=\ket{0}^{\otimes m}\bra{0}^{\otimes n}+a\ket{1}^{\otimes m}\bra{1}^{\otimes n},
%\left\llbracket \tikzfig{TikZit/generalgreenspider}\right\rrbracket =0^{\otimes m}+a^{\otimes m}, 
\]
\[
\left\llbracket %
	\beginpgfgraphicnamed{TikZit/generalredspider}
	\InputIfFileExists{TikZit/generalredspider.tikz}{}{\input{./figures/TikZit/generalredspider.tikz}}%
	\endpgfgraphicnamed
 \right\rrbracket=\ket{+}^{\otimes m}\bra{+}^{\otimes n}+a\ket{-}^{\otimes m}\bra{-}^{\otimes n},
\]
\[
\left\llbracket%
	\beginpgfgraphicnamed{TikZit/HadaDecomSingleslt}
	\begin{tikzpicture}
	\begin{pgfonlayer}{nodelayer}
		\node [style=H box] (0) at (-0.75, 0) {$H$};
		\node [style=none] (1) at (-0.75, -0.5) {};
		\node [style=none] (2) at (-0.75, 0.5) {};
	\end{pgfonlayer}
	\begin{pgfonlayer}{edgelayer}
		\draw (2.center) to (0);
		\draw (1.center) to (0);
	\end{pgfonlayer}
\end{tikzpicture}}%
	\endpgfgraphicnamed
\right\rrbracket=\frac{1}{\sqrt{2}}\begin{pmatrix}
        1 & 1 \\
        1 & -1
 \end{pmatrix}, \quad
  \left\llbracket%
	\beginpgfgraphicnamed{TikZit/emptysquare}
	\InputIfFileExists{TikZit/emptysquare.tikz}{}{\input{./figures/TikZit/emptysquare.tikz}}%
	\endpgfgraphicnamed
\right\rrbracket=1, \quad
\left\llbracket%
	\beginpgfgraphicnamed{TikZit/Id}
	\begin{tikzpicture}
	\begin{pgfonlayer}{nodelayer}
		\node [style=none] (1) at (0.5, 0.3) {};
		\node [style=none] (2) at (0.5, -0.3) {};
		\node [style=none] (3) at (0.5, -0.5) {};
		\node [style=none] (4) at (0.5, 0.5) {};
	\end{pgfonlayer}
	\begin{pgfonlayer}{edgelayer}
		\draw (1.center) to (2.center);
	\end{pgfonlayer}
\end{tikzpicture}}%
	\endpgfgraphicnamed
\right\rrbracket=\begin{pmatrix}
        1 & 0 \\
        0 & 1
 \end{pmatrix}, 
   \]
\[
 \left\llbracket%
	\beginpgfgraphicnamed{TikZit/triangle}
	\begin{tikzpicture}
	\begin{pgfonlayer}{nodelayer}
		\node [style=none] (0) at (0, 0.5) {};
		\node [style=triangle] (1) at (0, 0) {};
		\node [style=none] (2) at (0, -0.5) {};
	\end{pgfonlayer}
	\begin{pgfonlayer}{edgelayer}
		\draw (0.center) to (2.center);
	\end{pgfonlayer}
\end{tikzpicture}}%
	\endpgfgraphicnamed
\right\rrbracket=\begin{pmatrix}
        1 & 1 \\
        0 & 1
 \end{pmatrix}, \quad \quad
  \left\llbracket%
	\beginpgfgraphicnamed{TikZit/triangleinv}
	\begin{tikzpicture}
	\begin{pgfonlayer}{nodelayer}
		\node [style=none] (0) at (0.25, 0.25) {-{\scriptsize1}};
		\node [style=triangle] (1) at (0, 0) {};
		\node [style=none] (2) at (0, -0.5) {};
		\node [style=none] (3) at (0, 0.5) {};
	\end{pgfonlayer}
	\begin{pgfonlayer}{edgelayer}
		\draw (3.center) to (2.center);
	\end{pgfonlayer}
\end{tikzpicture}}%
	\endpgfgraphicnamed
\right\rrbracket=\begin{pmatrix}
        1 & -1 \\
        0 & 1
 \end{pmatrix}.
 \]

\[
 \left\llbracket%
	\beginpgfgraphicnamed{TikZit/swap}
	\InputIfFileExists{TikZit/swap.tikz}{}{\input{./figures/TikZit/swap.tikz}}%
	\endpgfgraphicnamed
\right\rrbracket=\begin{pmatrix}
        1 & 0 & 0 & 0 \\
        0 & 0 & 1 & 0 \\
        0 & 1 & 0 & 0 \\
        0 & 0 & 0 & 1 
 \end{pmatrix}, \quad
  \left\llbracket%
	\beginpgfgraphicnamed{TikZit/cap}
	\begin{tikzpicture}
	\begin{pgfonlayer}{nodelayer}
		\node [style=none] (0) at (0, -0) {};
		\node [style=none] (1) at (1, -0) {};
	\end{pgfonlayer}
	\begin{pgfonlayer}{edgelayer}
		\draw [bend left=90, looseness=1.50] (0.center) to (1.center);
	\end{pgfonlayer}
\end{tikzpicture}}%
	\endpgfgraphicnamed
\right\rrbracket=\begin{pmatrix}
        1  \\
        0  \\
        0  \\
        1  \\
 \end{pmatrix}, \quad
   \left\llbracket%
	\beginpgfgraphicnamed{TikZit/cup}
	\begin{tikzpicture}
	\begin{pgfonlayer}{nodelayer}
		\node [style=none] (0) at (0, 0.5) {};
		\node [style=none] (1) at (1, 0.5) {};
	\end{pgfonlayer}
	\begin{pgfonlayer}{edgelayer}
		\draw [bend right=90, looseness=1.50] (0.center) to (1.center);
	\end{pgfonlayer}
\end{tikzpicture}}%
	\endpgfgraphicnamed
\right\rrbracket=\begin{pmatrix}
        1 & 0 & 0 & 1 
         \end{pmatrix},
   \]

\[  \llbracket D_1\otimes D_2  \rrbracket =  \llbracket D_1  \rrbracket \otimes  \llbracket  D_2  \rrbracket, \quad 
 \llbracket D_1\circ D_2  \rrbracket =  \llbracket D_1  \rrbracket \circ  \llbracket  D_2  \rrbracket,
  \]
where 
$$ \ket{0}= \begin{pmatrix}
        1  \\
        0  \\
 \end{pmatrix}, \quad 
 \bra{0}=\begin{pmatrix}
        1 & 0 
         \end{pmatrix},
 \quad  \ket{1}= \begin{pmatrix}
        0  \\
        1  \\
 \end{pmatrix}, \quad 
  \bra{1}=\begin{pmatrix}
     0 & 1 
         \end{pmatrix},
 $$
$$ \ket{+}= \frac{1}{\sqrt{2}}\begin{pmatrix}
        1  \\
        1  \\
 \end{pmatrix}, \quad 
 \bra{+}=\frac{1}{\sqrt{2}}\begin{pmatrix}
        1 & 1
         \end{pmatrix},
 \quad  \ket{-}= \frac{1}{\sqrt{2}}\begin{pmatrix}
        1  \\
        -1  \\
 \end{pmatrix}, \quad 
  \bra{-}=\frac{1}{\sqrt{2}}\begin{pmatrix}
     1 & -1 
         \end{pmatrix}.
 $$

 Now we give a purely algebraic set of rules for ZX-calculus in the sense that there are no trigonometry functions such as $\sin$ and $\cos$ involved.

 \begin{figure}[!h]
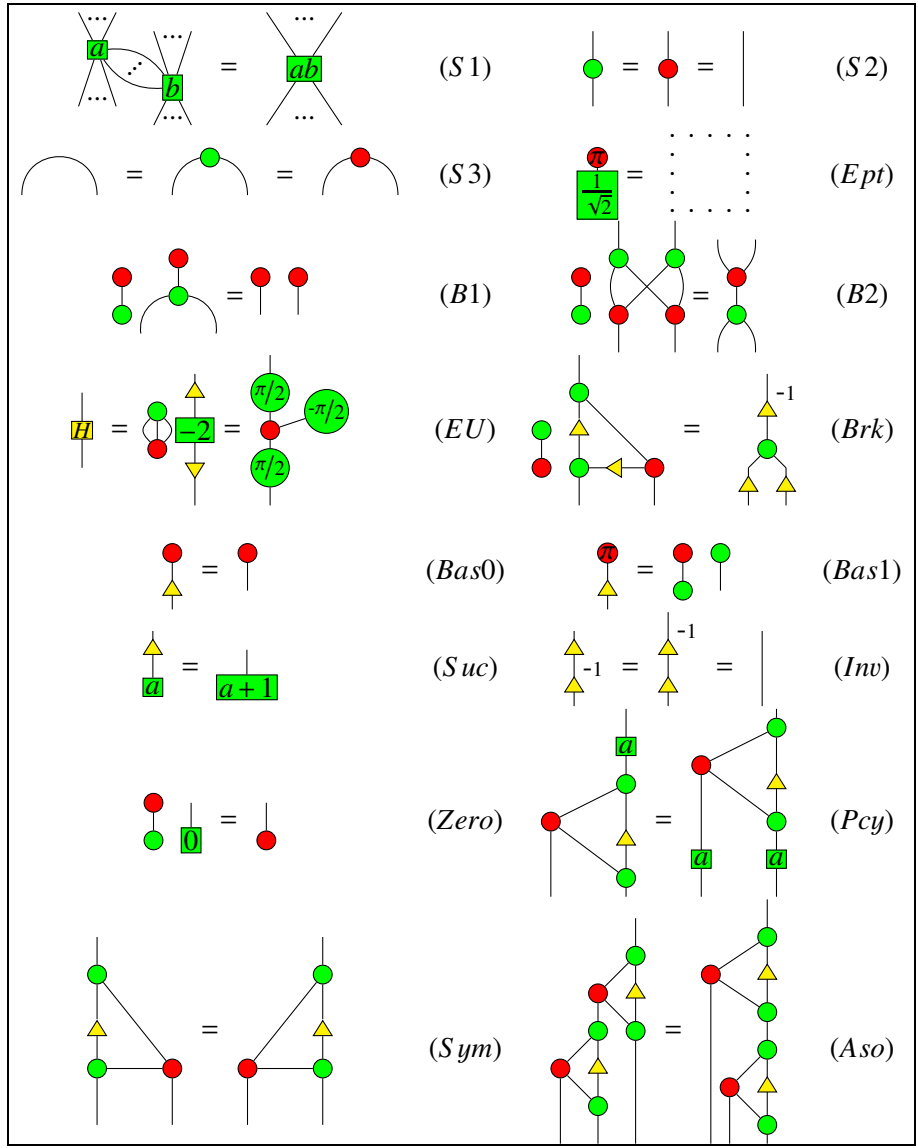

\begin{center}
\[
\quad \qquad\begin{array}{|cccc|}
\hline
	\beginpgfgraphicnamed{TikZit/generalgreenspiderfuse}
	\InputIfFileExists{TikZit/generalgreenspiderfuse.tikz}{}{\input{./figures/TikZit/generalgreenspiderfuse.tikz}}%
	\endpgfgraphicnamed
&(S1) &%
	\beginpgfgraphicnamed{TikZit/s2new2}
	\InputIfFileExists{TikZit/s2new2.tikz}{}{\input{./figures/TikZit/s2new2.tikz}}%
	\endpgfgraphicnamed
 &(S2)\\
	\beginpgfgraphicnamed{TikZit/induced_compact_structure}
	\InputIfFileExists{TikZit/induced_compact_structure.tikz}{}{\input{./figures/TikZit/induced_compact_structure.tikz}}%
	\endpgfgraphicnamed
&(S3) & %
	\beginpgfgraphicnamed{TikZit/rpirt2empt}
	\InputIfFileExists{TikZit/rpirt2empt.tikz}{}{\input{./figures/TikZit/rpirt2empt.tikz}}%
	\endpgfgraphicnamed
  &(Ept) \\
	\beginpgfgraphicnamed{TikZit/b1slt}
	\InputIfFileExists{TikZit/b1slt.tikz}{}{\input{./figures/TikZit/b1slt.tikz}}%
	\endpgfgraphicnamed
=%
	\beginpgfgraphicnamed{TikZit/b1srt}
	\InputIfFileExists{TikZit/b1srt.tikz}{}{\input{./figures/TikZit/b1srt.tikz}}%
	\endpgfgraphicnamed
&(B1)  & %
	\beginpgfgraphicnamed{TikZit/b2slt}
	\InputIfFileExists{TikZit/b2slt.tikz}{}{\input{./figures/TikZit/b2slt.tikz}}%
	\endpgfgraphicnamed
=%
	\beginpgfgraphicnamed{TikZit/b2srt}
	\InputIfFileExists{TikZit/b2srt.tikz}{}{\input{./figures/TikZit/b2srt.tikz}}%
	\endpgfgraphicnamed
&(B2)\\ 
% \tikzfig{TikZit/redpicopy}  &(Picp)&  \tikzfig{TikZit/generalredspiderdefini2}&(H) \\
  %
	\beginpgfgraphicnamed{TikZit/eu2styles2}
	\InputIfFileExists{TikZit/eu2styles2.tikz}{}{\input{./figures/TikZit/eu2styles2.tikz}}%
	\endpgfgraphicnamed
&(EU) & %
	\beginpgfgraphicnamed{TikZit/anddflip}
	\InputIfFileExists{TikZit/anddflip.tikz}{}{\input{./figures/TikZit/anddflip.tikz}}%
	\endpgfgraphicnamed
&(Brk) \\
 & &&\\
	\beginpgfgraphicnamed{TikZit/triangleocopy}
	\InputIfFileExists{TikZit/triangleocopy.tikz}{}{\input{./figures/TikZit/triangleocopy.tikz}}%
	\endpgfgraphicnamed
 &(Bas0) &%
	\beginpgfgraphicnamed{TikZit/trianglepicopy}
	\InputIfFileExists{TikZit/trianglepicopy.tikz}{}{\input{./figures/TikZit/trianglepicopy.tikz}}%
	\endpgfgraphicnamed
&(Bas1)\\
	\beginpgfgraphicnamed{TikZit/plus1}
	\InputIfFileExists{TikZit/plus1.tikz}{}{\input{./figures/TikZit/plus1.tikz}}%
	\endpgfgraphicnamed
&(Suc)& %
	\beginpgfgraphicnamed{TikZit/triangleinvers}
	\InputIfFileExists{TikZit/triangleinvers.tikz}{}{\input{./figures/TikZit/triangleinvers.tikz}}%
	\endpgfgraphicnamed
  & (Inv) \\
	\beginpgfgraphicnamed{TikZit/zerotored}
	\InputIfFileExists{TikZit/zerotored.tikz}{}{\input{./figures/TikZit/zerotored.tikz}}%
	\endpgfgraphicnamed
&(Zero)&%
	\beginpgfgraphicnamed{TikZit/TR1314combine2}
	\InputIfFileExists{TikZit/TR1314combine2.tikz}{}{\input{./figures/TikZit/TR1314combine2.tikz}}%
	\endpgfgraphicnamed
 &(Pcy)\\
	\beginpgfgraphicnamed{TikZit/lemma4}
	\InputIfFileExists{TikZit/lemma4.tikz}{}{\input{./figures/TikZit/lemma4.tikz}}%
	\endpgfgraphicnamed
&(Sym) &  %
	\beginpgfgraphicnamed{TikZit/associate}
	\InputIfFileExists{TikZit/associate.tikz}{}{\input{./figures/TikZit/associate.tikz}}%
	\endpgfgraphicnamed
 &(Aso)\\ 
%\tikzfig{TikZit/anddflip}&(Brk) &&\\ 
  		  		\hline
  		\end{array}\]   
  	\end{center}
  	\caption{Algebraic rules, $a, b \in \mathbb C.$}\label{figurealgebra}
  \end{figure}
 \FloatBarrier
 \begin{remark}
The last three rules are all about the properties of the W state  %
	\beginpgfgraphicnamed{TikZit/wstatezx}
	\InputIfFileExists{TikZit/wstatezx.tikz}{}{\input{./figures/TikZit/wstatezx.tikz}}%
	\endpgfgraphicnamed
: (Pcy) means  phase copy,  i.e., any phase can be copied by the W state;  (Sym) means symmetry, i.e., the W state is symmetric; (Aso) means associativity, i.e., the W state is associative. In addition, the rules (H) and (S1) are first introduced in \cite{ngwang2}.
 \end{remark}

It is a routine check that these rules are sound in the sense that they still hold under the standard interpretation $\left\llbracket \cdot \right\rrbracket$. We mention again that a significant application of these  algebraic rules is the derivation of the so-called spider nest identities \cite{bobanthonywang}, which are key to the T-count reduction of quantum circuits \cite{nielbianwangtqc, nielbianwang}. 

With the standard interpretation and the above rules, we can define the completeness of ZX-calculus. 
  \begin{definition}
 The ZX-calculus is called complete if for  any two diagrams $D_1$ and $D_2$,  $\llbracket D_1  \rrbracket =  \llbracket  D_2  \rrbracket$ must imply that $ZX\vdash D_1=D_2$.
  \end{definition}
 % \subsection*{Discussion of the ZX rules}
   \begin{remark}
All the rules in Figure \ref{figurealgebra}  now are  algebraic, and they will be proved to be a complete axiomatisation of ZX-calculs. One may wonder how these rules are obtained.  The answer is simply that giving rules is actually a constructive thing, and these rules are basically refined from plenty of practice of diagrammatical rewriting based on the previous rules \cite{amarngwanglics}. The reason why trigonometry functions can be eliminated in this new set of rules is because the generators we have chosen in Table  \ref{qbzxgenerator} allow an algebraic translation (an isomorphic functor in fact) from the ZW-calculus which is algebraic at the beginning \cite{amarngwanglics}. Furthermore, we note that this new set of algebraic rules are not unique, as one could add more rules to the set or derive equivalent rules from them. The most important thing for choosing a set of rules is that they should be useful in applications as much as possible.
 \end{remark}

  Below we give some useful properties following from  Figure \ref{figurealgebra}.
  \begin{lemma}
%\begin{equation}\label{lem:hopf}
	 %
	\beginpgfgraphicnamed{TikZit/lemmahopf}
	\InputIfFileExists{TikZit/lemmahopf.tikz}{}{\input{./figures/TikZit/lemmahopf.tikz}}%
	\endpgfgraphicnamed
 (Hopf)
	%\end{equation}
\end{lemma}
This has been proved in \cite{bpw} based on rules (S1), (S2), (S3), (B1), (B2) and the definition of red spider.
  
   \begin{lemma}\label{Invscalar}
	\beginpgfgraphicnamed{TikZit/lemmainv}
	\InputIfFileExists{TikZit/lemmainv.tikz}{}{\input{./figures/TikZit/lemmainv.tikz}}%
	\endpgfgraphicnamed
  (Ivs)
     \end{lemma}
\begin{proof}
	\beginpgfgraphicnamed{TikZit/origininvprf}
	\InputIfFileExists{TikZit/origininvprf.tikz}{}{\input{./figures/TikZit/origininvprf.tikz}}%
	\endpgfgraphicnamed
  
 \end{proof}
  \begin{lemma}
	\beginpgfgraphicnamed{TikZit/lemma6b}
	\InputIfFileExists{TikZit/lemma6b.tikz}{}{\input{./figures/TikZit/lemma6b.tikz}}%
	\endpgfgraphicnamed
 (Picp)
     \end{lemma}
  This has been proved in \cite{bpw} based on rules (S1), (S2), (S3), (B1), (B2) and the definition of red spider.
  
  \begin{lemma}
	\beginpgfgraphicnamed{TikZit/lemmacom}
	\InputIfFileExists{TikZit/lemmacom.tikz}{}{\input{./figures/TikZit/lemmacom.tikz}}%
	\endpgfgraphicnamed
 (Com)
\end{lemma}
	  This has been proved in \cite{bpw} based on rules (S1), (S2), (S3), (B1), (B2) and the definition of red spider. 
     
  \begin{lemma}\label{redpigdot}
	\beginpgfgraphicnamed{TikZit/redpigrdot}
	\InputIfFileExists{TikZit/redpigrdot.tikz}{}{\input{./figures/TikZit/redpigrdot.tikz}}%
	\endpgfgraphicnamed
 
\end{lemma}
\begin{proof}
	\beginpgfgraphicnamed{TikZit/redpigrdotprf}
	\InputIfFileExists{TikZit/redpigrdotprf.tikz}{}{\input{./figures/TikZit/redpigrdotprf.tikz}}%
	\endpgfgraphicnamed
  
 \end{proof}     
     
\section{Proof of completeness}
In this section, we prove that the rules in Figure \ref{figurealgebra} are complete for ZX-calculus. Since it is already proved in \cite{amarngwanglics} that ZX-calculus is complete with the rules presented in Figure  \ref{figureold1} and Figure  \ref{figureold2}, we only need to prove that all the rules in Figure  \ref{figureold1} and  \ref{figureold2} can be derived from rules in Figure  \ref{figurealgebra}.

\begin{figure}[!h]
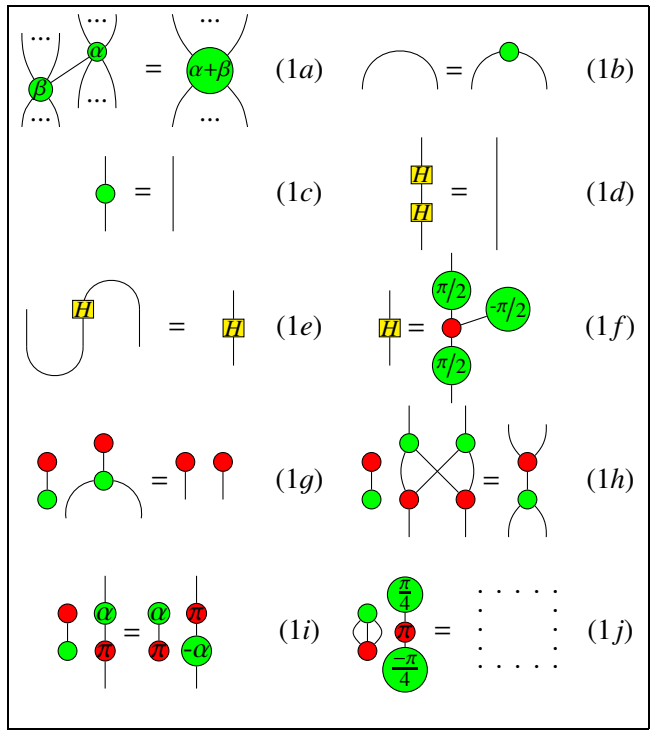

\begin{center}
\[
\quad \qquad\begin{array}{|cccc|}
\hline
	\beginpgfgraphicnamed{TikZit/spiderg1a}
	\InputIfFileExists{TikZit/spiderg1a.tikz}{}{\input{./figures/TikZit/spiderg1a.tikz}}%
	\endpgfgraphicnamed
 &(1a) & %
	\beginpgfgraphicnamed{TikZit/induced_compact_structure-2wirelt}
	\InputIfFileExists{TikZit/induced_compact_structure-2wirelt.tikz}{}{\input{./figures/TikZit/induced_compact_structure-2wirelt.tikz}}%
	\endpgfgraphicnamed
=%
	\beginpgfgraphicnamed{TikZit/induced_compact_structure-2wirert}
	\InputIfFileExists{TikZit/induced_compact_structure-2wirert.tikz}{}{\input{./figures/TikZit/induced_compact_structure-2wirert.tikz}}%
	\endpgfgraphicnamed
&(1b)\\
	\beginpgfgraphicnamed{TikZit/s2new}
	\InputIfFileExists{TikZit/s2new.tikz}{}{\input{./figures/TikZit/s2new.tikz}}%
	\endpgfgraphicnamed
 &(1c) & %
	\beginpgfgraphicnamed{TikZit/hsquare}
	\InputIfFileExists{TikZit/hsquare.tikz}{}{\input{./figures/TikZit/hsquare.tikz}}%
	\endpgfgraphicnamed
 &(1d)\\
	\beginpgfgraphicnamed{TikZit/hslidecap1e}
	\InputIfFileExists{TikZit/hslidecap1e.tikz}{}{\input{./figures/TikZit/hslidecap1e.tikz}}%
	\endpgfgraphicnamed
 &(1e) &  %
	\beginpgfgraphicnamed{TikZit/HadaDecomSingleslt}
	\InputIfFileExists{TikZit/HadaDecomSingleslt.tikz}{}{\input{./figures/TikZit/HadaDecomSingleslt.tikz}}%
	\endpgfgraphicnamed
= %
	\beginpgfgraphicnamed{TikZit/HadaDecomSinglesrt2}
	\InputIfFileExists{TikZit/HadaDecomSinglesrt2.tikz}{}{\input{./figures/TikZit/HadaDecomSinglesrt2.tikz}}%
	\endpgfgraphicnamed
&(1f)\\
%&\tikzfig{TikZit/hslidecap} &(H3) &\\
    %
	\beginpgfgraphicnamed{TikZit/b1slt}
	\InputIfFileExists{TikZit/b1slt.tikz}{}{\input{./figures/TikZit/b1slt.tikz}}%
	\endpgfgraphicnamed
=%
	\beginpgfgraphicnamed{TikZit/b1srt}
	\InputIfFileExists{TikZit/b1srt.tikz}{}{\input{./figures/TikZit/b1srt.tikz}}%
	\endpgfgraphicnamed
&(1g) & %
	\beginpgfgraphicnamed{TikZit/b2slt}
	\InputIfFileExists{TikZit/b2slt.tikz}{}{\input{./figures/TikZit/b2slt.tikz}}%
	\endpgfgraphicnamed
=%
	\beginpgfgraphicnamed{TikZit/b2srt}
	\InputIfFileExists{TikZit/b2srt.tikz}{}{\input{./figures/TikZit/b2srt.tikz}}%
	\endpgfgraphicnamed
&(1h)\\
    &&&\\ 
	\beginpgfgraphicnamed{TikZit/k2slt}
	\InputIfFileExists{TikZit/k2slt.tikz}{}{\input{./figures/TikZit/k2slt.tikz}}%
	\endpgfgraphicnamed
=%
	\beginpgfgraphicnamed{TikZit/k2srt}
	\InputIfFileExists{TikZit/k2srt.tikz}{}{\input{./figures/TikZit/k2srt.tikz}}%
	\endpgfgraphicnamed
&(1i) &   %
	\beginpgfgraphicnamed{TikZit/newemptyrl}
	\InputIfFileExists{TikZit/newemptyrl.tikz}{}{\input{./figures/TikZit/newemptyrl.tikz}}%
	\endpgfgraphicnamed
 &(1j)\\
&&&\\ 
\hline
\end{array}\]
\end{center}
  \caption{Previous ZX-calculus rules I, where $\alpha, \beta\in [0,~2\pi)$.}\label{figureold1}  
  \end{figure}
  \FloatBarrier

 \begin{figure}[!h]
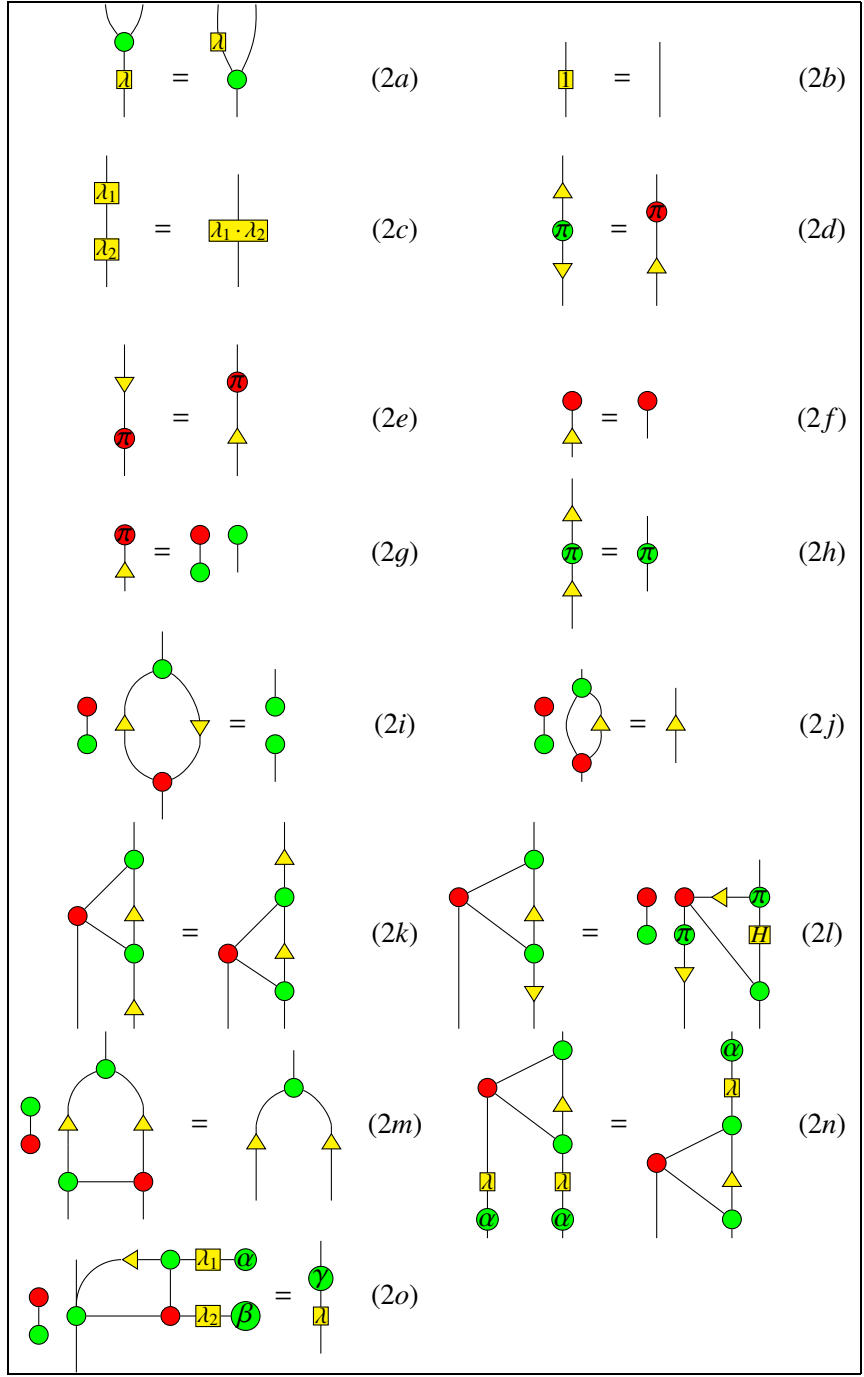

  	\begin{center}
  		\[
  		\quad \qquad\begin{array}{|cccc|}
  		\hline
	\beginpgfgraphicnamed{TikZit/lambbranch2}
	\InputIfFileExists{TikZit/lambbranch2.tikz}{}{\input{./figures/TikZit/lambbranch2.tikz}}%
	\endpgfgraphicnamed
&(2a)&%
	\beginpgfgraphicnamed{TikZit/sqr1is1}
	\InputIfFileExists{TikZit/sqr1is1.tikz}{}{\input{./figures/TikZit/sqr1is1.tikz}}%
	\endpgfgraphicnamed
&(2b) \\
		 &&&\\ 
	\beginpgfgraphicnamed{TikZit/lambdatimes}
	\InputIfFileExists{TikZit/lambdatimes.tikz}{}{\input{./figures/TikZit/lambdatimes.tikz}}%
	\endpgfgraphicnamed
&(2c)&%
	\beginpgfgraphicnamed{TikZit/gpiintriangles}
	\InputIfFileExists{TikZit/gpiintriangles.tikz}{}{\input{./figures/TikZit/gpiintriangles.tikz}}%
	\endpgfgraphicnamed
&(2d) \\
		 &&&\\ 
	\beginpgfgraphicnamed{TikZit/trianglepicommute}
	\InputIfFileExists{TikZit/trianglepicommute.tikz}{}{\input{./figures/TikZit/trianglepicommute.tikz}}%
	\endpgfgraphicnamed
 &(2e) &%
	\beginpgfgraphicnamed{TikZit/triangleocopy}
	\InputIfFileExists{TikZit/triangleocopy.tikz}{}{\input{./figures/TikZit/triangleocopy.tikz}}%
	\endpgfgraphicnamed
 &(2f)\\

	\beginpgfgraphicnamed{TikZit/trianglepicopy}
	\InputIfFileExists{TikZit/trianglepicopy.tikz}{}{\input{./figures/TikZit/trianglepicopy.tikz}}%
	\endpgfgraphicnamed
&(2g) & %
	\beginpgfgraphicnamed{TikZit/tr10prime}
	\InputIfFileExists{TikZit/tr10prime.tikz}{}{\input{./figures/TikZit/tr10prime.tikz}}%
	\endpgfgraphicnamed
 &(2h)\\
	\beginpgfgraphicnamed{TikZit/tr5prime}
	\InputIfFileExists{TikZit/tr5prime.tikz}{}{\input{./figures/TikZit/tr5prime.tikz}}%
	\endpgfgraphicnamed
&(2i) & %
	\beginpgfgraphicnamed{TikZit/trianglehopf2}
	\InputIfFileExists{TikZit/trianglehopf2.tikz}{}{\input{./figures/TikZit/trianglehopf2.tikz}}%
	\endpgfgraphicnamed
&(2j)\\
	\beginpgfgraphicnamed{TikZit/2triangleup}
	\InputIfFileExists{TikZit/2triangleup.tikz}{}{\input{./figures/TikZit/2triangleup.tikz}}%
	\endpgfgraphicnamed
&(2k) & %
	\beginpgfgraphicnamed{TikZit/2triangledown}
	\InputIfFileExists{TikZit/2triangledown.tikz}{}{\input{./figures/TikZit/2triangledown.tikz}}%
	\endpgfgraphicnamed
&(2l)\\
	\beginpgfgraphicnamed{TikZit/2triangledeloopnopi}
	\InputIfFileExists{TikZit/2triangledeloopnopi.tikz}{}{\input{./figures/TikZit/2triangledeloopnopi.tikz}}%
	\endpgfgraphicnamed
&(2m)&%
	\beginpgfgraphicnamed{TikZit/TR1314combine}
	\InputIfFileExists{TikZit/TR1314combine.tikz}{}{\input{./figures/TikZit/TR1314combine.tikz}}%
	\endpgfgraphicnamed
 &(2n) \\
	\beginpgfgraphicnamed{TikZit/plusnew}
	\InputIfFileExists{TikZit/plusnew.tikz}{}{\input{./figures/TikZit/plusnew.tikz}}%
	\endpgfgraphicnamed
 &(2o)&&\\
  		\hline
  		\end{array}\]
  	\end{center}
  	
  	\caption{Previous ZX-calculus rules II, where  $\lambda, \lambda_1,  \lambda_2 \geq 0, \alpha, \beta, \gamma \in [0,~2\pi);$ in (2o), $\lambda e^{i\gamma} 
  =\lambda_1 e^{i\alpha}+ \lambda_2 e^{i\beta}$.}\label{figureold2}
  \end{figure}
 \FloatBarrier
   %Now all the rules listed in Figure  \ref{figureold1} and Figure  \ref{figureold2} have been derived from Figure  \ref{figurealgebra}. So we have
    %or included 
    %in Figure 1 and 2, which means all the rules in  Figure 1 and 2 are complete, while  the rules in  Figure 3 should be derivable.
    \begin{remark}
  In the rule (2o), the equality  $\lambda e^{i\gamma} 
  =\lambda_1 e^{i\alpha}+ \lambda_2 e^{i\beta}$ is not expressed in terms of trigonometry functions, but if we solve this equation to give the relations between the angles, then we get 
  $\lambda=\sqrt{\lambda_1^2+\lambda_2^2+2\lambda_1\lambda_2\cos(\alpha-\beta)}, ~\gamma=\alpha+\arccos(\frac{\lambda_1+\lambda_2\cos(\alpha-\beta)}{\lambda})$, if $\lambda\neq 0;$   and $\alpha=\beta+k\pi, ~|\lambda_1| = |\lambda_2|,$ $\gamma$ can be any angle, if $\lambda=0$.

        \end{remark}

    \begin{remark}
The main difference between the algebraic axiomatisation in this paper and the previous axiomatisation in  \cite{amarngwanglics} lies in that we use as generator a green box with parameters ranging over complex numbers in the algebraic axiomatisation rather then a yellow box with parameters ranging over non-negative real numbers in the previous axiomatisation. Furthermore, the rules with triangles involved in Figure \ref{figurealgebra}  are simpler and more natural than those with triangles involved as shown in Figure \ref{figureold2}. The usefulness of the new generator green box has been fully shown in  \cite{coeckewang} through the derivation of the (P) rule with the main help from the green box.

 \end{remark}

       \begin{theorem}\label{zxalgcomplete}
       ZX-calculus is complete for pure qubit quantum mechanics with the rules listed in Figure  \ref{figurealgebra}. 
        \end{theorem}
 The proof is given in the appendix.

%\fi

%\iffalse
\section{From ZH-calculus to ZX-calculus}
ZH-calculus is another graphical language for quantum computing introduced by Backens and Kissinger  \cite{miriamaleks}. It has found applications in various fields like tensor network \cite{nakons2020} and hypergraph states  \cite{lwaleks2020}. So it would be very useful if we could establish a connection between ZH-calculus and ZX-calculus. In \cite{Wetering}, there are translations established between phase-free ZH-calculus and ZX-calculus, yet there has been no translation found between the full ZH-calculus and ZX-calculus. In this section, with the algebraic axiomatisation of ZX-calculus, we are able to give a simple translation from any ZH diagrams to ZX diagrams with the semantics preserved. Furthermore, we derive all the translated ZH rules within ZX by the algebraic rules given in Figure  \ref{figurealgebra}. Although in principle we know this can be done because of the completeness of ZX-calculus, but that is not a constructive way.  We give the details of all such derivations which is far from trivial especially for the last three ZH rules, which means the translation from ZH diagrams to ZX diagrams alone doesn't guarantee an easy derivation of the ZH rules. By these translation and rule-derivation we have the bonus that any result obtained via ZH-calculus can be transplanted to  ZX-calculus.

The ZH-calculus is also based on a PROP, thus can be presented by generators and rewriting rules. First we list its generators as follows \cite{miriamaleks}:
\begin{table}[!h]
\begin{center} 
\begin{tabular}{|r@{~}r@{~}c@{~}c|r@{~}r@{~}c@{~}c|}
\hline
$Z^{(n,m)}$&$:$&$n\to m$ & %
	\beginpgfgraphicnamed{TikZit/zhwhitespider}
	\InputIfFileExists{TikZit/zhwhitespider.tikz}{}{\input{./figures/TikZit/zhwhitespider.tikz}}%
	\endpgfgraphicnamed
  & $H_{a}^{(n,m)}$&$:$&$n\to m$ & %
	\beginpgfgraphicnamed{TikZit/hspider}
	\InputIfFileExists{TikZit/hspider.tikz}{}{\input{./figures/TikZit/hspider.tikz}}%
	\endpgfgraphicnamed
  \\\hline
   $\mathbb I$&$:$&$1\to 1$&%
	\beginpgfgraphicnamed{TikZit/Id}
	\InputIfFileExists{TikZit/Id.tikz}{}{\input{./figures/TikZit/Id.tikz}}%
	\endpgfgraphicnamed
 &  $\sigma$&$:$&$ 2\to 2$& %
	\beginpgfgraphicnamed{TikZit/swap}
	\InputIfFileExists{TikZit/swap.tikz}{}{\input{./figures/TikZit/swap.tikz}}%
	\endpgfgraphicnamed
\\\hline
   $C_a$&$:$&$ 0\to 2$& %
	\beginpgfgraphicnamed{TikZit/cap}
	\InputIfFileExists{TikZit/cap.tikz}{}{\input{./figures/TikZit/cap.tikz}}%
	\endpgfgraphicnamed
 &$ C_u$&$:$&$ 2\to 0$&%
	\beginpgfgraphicnamed{TikZit/cup}
	\InputIfFileExists{TikZit/cup.tikz}{}{\input{./figures/TikZit/cup.tikz}}%
	\endpgfgraphicnamed
 \\\hline
%  $\sigma$&$:$&$ 2\to 2$& \tikzfig{TikZit/swap}&&&&\\\hline

\end{tabular}\caption{Generators of ZH-calculus,where $m,n\in \mathbb N$, $ a  \in \mathbb C$.} \label{qbzhgenerator}
\end{center}
\end{table}
 \FloatBarrier
     
There is also a standard interpretation $\left\llbracket \cdot \right\rrbracket$ for the ZH diagrams, here we only present the interpretation of the first two generators, as other generators are the same as that of ZX-calculus:
\[
\left\llbracket %
	\beginpgfgraphicnamed{TikZit/zhwhitespider}
	\InputIfFileExists{TikZit/zhwhitespider.tikz}{}{\input{./figures/TikZit/zhwhitespider.tikz}}%
	\endpgfgraphicnamed
 \right\rrbracket=\ket{0}^{\otimes m}\bra{0}^{\otimes n}+\ket{1}^{\otimes m}\bra{1}^{\otimes n},
%\left\llbracket \tikzfig{TikZit/generalgreenspider}\right\rrbracket =0^{\otimes m}+a^{\otimes m}, 
\]
\[
\left\llbracket %
	\beginpgfgraphicnamed{TikZit/hspider}
	\InputIfFileExists{TikZit/hspider.tikz}{}{\input{./figures/TikZit/hspider.tikz}}%
	\endpgfgraphicnamed
 \right\rrbracket=\sum a^{i_1\cdots i_m j_1\cdots j_n}\ket{i_1\cdots i_m}\bra{j_1\cdots j_n}.
\]     
It is clear that the white spider in ZH-calculus is just the phase free green spider in ZX-calculus.   Thus we can give a semantics-preserving translation    
 $ \left\llbracket \cdot \right\rrbracket_{HX}$  from ZH-calculus to ZX-calculus via the translation of H-box:
 \[
 \left\llbracket %
	\beginpgfgraphicnamed{TikZit/hspider}
	\InputIfFileExists{TikZit/hspider.tikz}{}{\input{./figures/TikZit/hspider.tikz}}%
	\endpgfgraphicnamed
 \right\rrbracket_{HX}
=%
	\beginpgfgraphicnamed{TikZit/zhtozxgenerphasert}
	\InputIfFileExists{TikZit/zhtozxgenerphasert.tikz}{}{\input{./figures/TikZit/zhtozxgenerphasert.tikz}}%
	\endpgfgraphicnamed
 
\]
  
  This translation is totally new. Following  \cite{miriamaleks}, we make the convention that  
 $$ %
	\beginpgfgraphicnamed{TikZit/hspiderconcention}
	\InputIfFileExists{TikZit/hspiderconcention.tikz}{}{\input{./figures/TikZit/hspiderconcention.tikz}}%
	\endpgfgraphicnamed
,$$ then the ZH rules are given as follows:
     
   \begin{figure}[!h]
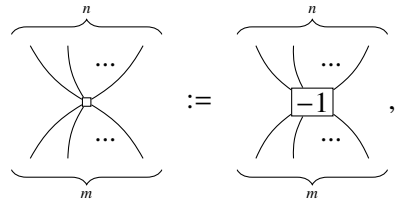
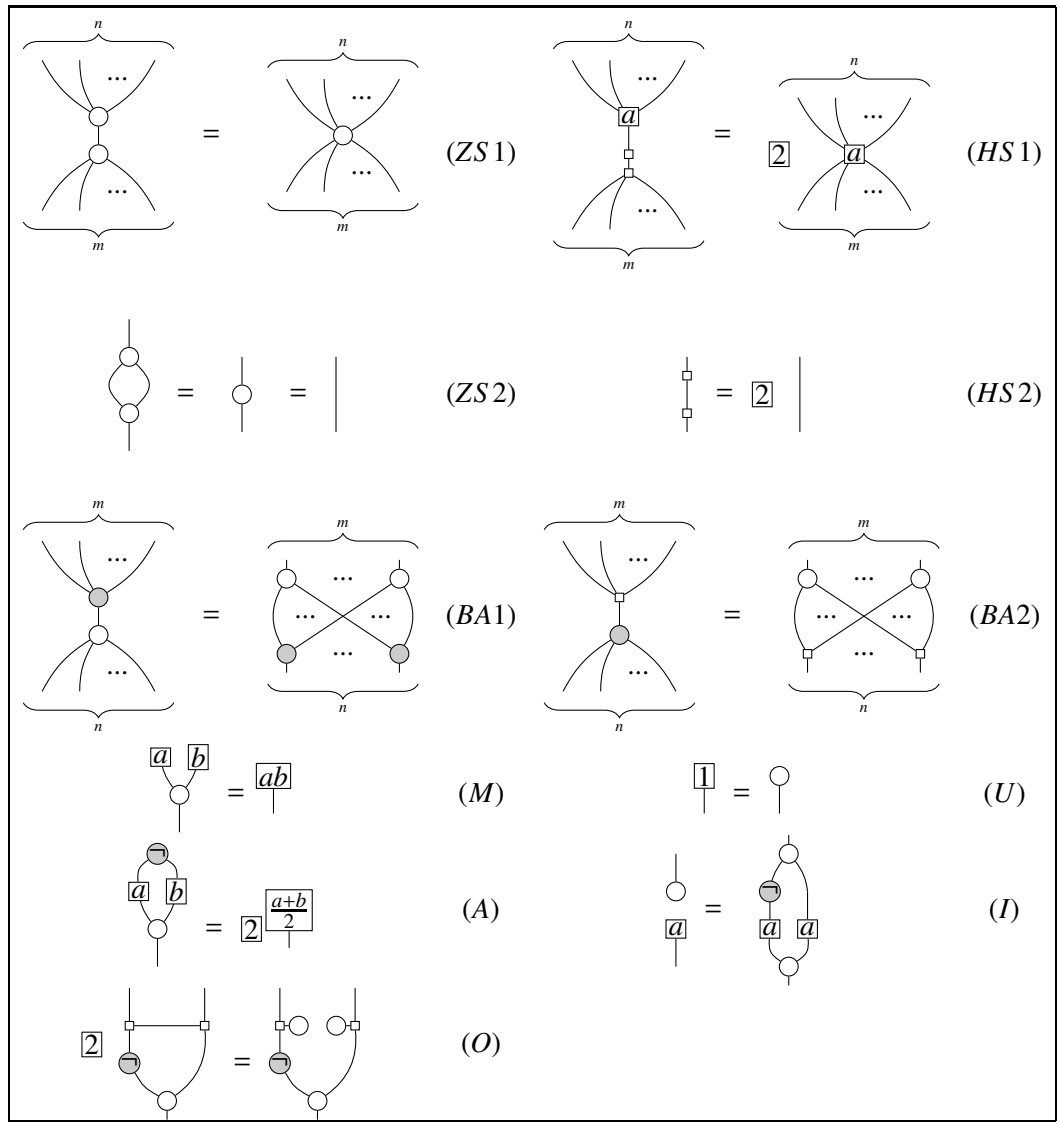

  	\begin{center}
  		\[
  		\quad \qquad\begin{array}{|cccc|}
  		\hline
	\beginpgfgraphicnamed{TikZit/zhzs1}
	\InputIfFileExists{TikZit/zhzs1.tikz}{}{\input{./figures/TikZit/zhzs1.tikz}}%
	\endpgfgraphicnamed
&(ZS1)&%
	\beginpgfgraphicnamed{TikZit/zhhs1}
	\InputIfFileExists{TikZit/zhhs1.tikz}{}{\input{./figures/TikZit/zhhs1.tikz}}%
	\endpgfgraphicnamed
&(HS1) \\
		 &&&\\ 
	\beginpgfgraphicnamed{TikZit/zhzs2}
	\InputIfFileExists{TikZit/zhzs2.tikz}{}{\input{./figures/TikZit/zhzs2.tikz}}%
	\endpgfgraphicnamed
&(ZS2)&%
	\beginpgfgraphicnamed{TikZit/zhhs2}
	\InputIfFileExists{TikZit/zhhs2.tikz}{}{\input{./figures/TikZit/zhhs2.tikz}}%
	\endpgfgraphicnamed
&(HS2) \\
		 &&&\\ 
	\beginpgfgraphicnamed{TikZit/zhba1}
	\InputIfFileExists{TikZit/zhba1.tikz}{}{\input{./figures/TikZit/zhba1.tikz}}%
	\endpgfgraphicnamed
 &(BA1) &%
	\beginpgfgraphicnamed{TikZit/zhba2}
	\InputIfFileExists{TikZit/zhba2.tikz}{}{\input{./figures/TikZit/zhba2.tikz}}%
	\endpgfgraphicnamed
 &(BA2)\\

	\beginpgfgraphicnamed{TikZit/zhm}
	\InputIfFileExists{TikZit/zhm.tikz}{}{\input{./figures/TikZit/zhm.tikz}}%
	\endpgfgraphicnamed
&(M) & %
	\beginpgfgraphicnamed{TikZit/zhu}
	\InputIfFileExists{TikZit/zhu.tikz}{}{\input{./figures/TikZit/zhu.tikz}}%
	\endpgfgraphicnamed
 &(U)\\
	\beginpgfgraphicnamed{TikZit/zha}
	\InputIfFileExists{TikZit/zha.tikz}{}{\input{./figures/TikZit/zha.tikz}}%
	\endpgfgraphicnamed
&(A) & %
	\beginpgfgraphicnamed{TikZit/zhi}
	\InputIfFileExists{TikZit/zhi.tikz}{}{\input{./figures/TikZit/zhi.tikz}}%
	\endpgfgraphicnamed
&(I)\\
	\beginpgfgraphicnamed{TikZit/zho}
	\InputIfFileExists{TikZit/zho.tikz}{}{\input{./figures/TikZit/zho.tikz}}%
	\endpgfgraphicnamed
&(O) &&\\
  		\hline
  		\end{array}\]
  	\end{center}
  	
  	\caption{ZH rules, where  $a, b  \in \mathbb C$.}\label{figurezhrules}
  \end{figure}
 \FloatBarrier   
 
 Now we derive the ZX-translated rules in Figure \ref{figurezhrules} from ZX rules.
 
 The rules (ZS1), (ZS2), (M) and (U) just follow directly from the ZX rules (S1) and (S2).
 The rule (HS2) follows directly the ZX rule (S2) and the definition (H) of red spider. The ZH rule (BA1) is just a generalisation of the ZX rule (B2), which has been proved in ZX papers, for example \cite{duncan_graph_2009}. So we only need to derive the remaining ZH rules (HS1), (A) (I) and (O) individually. This has been shown in the appendix.
 
   Therefore, we have 
     \begin{theorem}\label{zhcomplete} 
     All the ZX-translated ZH rules can be derived from ZX rules.
           \end{theorem}  
           
  \begin{remark}
  Obviously,  the completeness of ZX-calculus already implies that all the ZX-translated ZH rules can be derived from ZX rules. However, that is not a constructive proof, so we have no idea on how such derivation really happens. The consequence is that it does no help to ZX-calculus even if there is a great result obtained in the ZH-calculus. In another word, to know the conclusion that all the ZX-translated ZH rules can be derived from ZX rules is not enough, we need to show the details of the derivations. 
  \end{remark}
 \section{Conclusion and further work  }       
   In this paper, we give a purely algebraic axiomatisation of ZX-calculus by introducing new generators. We show the proof of completeness by deriving complete rules established previously. Based one this algebraic axiomatisation, we obtain a simple translation of diagrams from ZH-calculus to ZX-calculus, and derive all the ZX-translated ZH rules within ZX-calculus.
   
   In the next step, we would like to have a proof of completeness based on the algebraic rules presented in this paper via a normal form, rather than a translation from other graphical language. It is also interesting to go for another direction: translate diagrams from ZX-calculus to ZH-calculus and derive the ZX rules within ZH-calculus. Finally, it is worthwhile to apply these algebraic rules to the problem of quantum circuit optimisation.
 % To prove the rules of ZH-calculus, we need the following two properties which can be verified by plugging with red zero and pi nodes. Therefore, they can be derived from rules in Figure 1 and 2 because of completeness.

% By using this version of ZX, all ZH diagrams are very easy to be translated to ZX diagrams with the symbols of triangle and general phase spider, since the H-box can be translated as follows:
%\[
%\tikzfig{TikZit/zhtozxgenerphase2} 
%\]
%\tikzfig{TikZit/trybox} 

%This translation means the H-box has brought us nothing more than the triangle and general phases. In an other word, ZH is  equivalent to ZX up to local triangle operators.
%The relation between ZX and ZH is much closer than the relation between ZX and ZW. 
%For quantum computing, in the sense of 
%complementarity (analogue of orthogonality),  ZX is like a rectangular  Cartesian coordinate system, while ZW is like a Polar coordinate system (one with phase and the other without),  and ZH  a coordinate rotation or translation of ZX.
 
  %\fi
 
 %----------------------------------------------------------------------------

%   \iffalse \fi
 \section*{Acknowledgements} 

This work was supported by AFOSR grant FA2386-18-1-4028.
The author would like to thank Aleks Kissinger for showing him the proof of the distribution rule in ZH-calculus. The author thanks  Bob Coecke,  Niel de Beaudrap and Konstantinos Meichanetzidis for useful discussions on the title of this paper.  
%\scalebox{0.3}{ \tikzfig{TikZit/mod5_4_tpar}}

%\scalebox{0.3}{ \tikzfig{TikZit/mod5_4_before}}
\bibliographystyle{eptcs}
\bibliography{generic}

%\newpage
% \iffalse
\section*{Appendix: Propositions, Lemmas and Proofs}
\section*{Proof of Theorem \ref{zxalgcomplete}}
Below we use the rules from Figure  \ref{figurealgebra} to derive the rules in Figure  \ref{figureold1} and  \ref{figureold2} one by one.

First (1a), (1b) and (1c) follow directly from (S1), (S3) and  (S2) respectively. 
(1d) follows clearly from (S2) and the definition of red spider.  (1e) follows from (1d) and (S3). (1f) is just a part of (EU). (1g) and (1h) are exactly (B1) and (B2) respectively.

% \section{Derivable rules}
% \iffalse
% Talk about the symmetry of spiders and the Hopf law, which are very useful.
%Stress that we work in a compact closed category.
% \begin{itemize}
% \item  \tikzfig{TikZit/redpidotcopy} \quad  \quad  \quad  \tikzfig{TikZit/hsquare}  \\
%  Proved in \textcolor{purple}{[Backens, Perdrix, Wang, QPL'16]} \vspace{0.5cm}
%Directly derived from rules (B1), (Picp) and (S2) respectively. 
 %   \end{itemize}

\begin{proposition}\label{1iprf}
	\beginpgfgraphicnamed{TikZit/k2slt}
	\InputIfFileExists{TikZit/k2slt.tikz}{}{\input{./figures/TikZit/k2slt.tikz}}%
	\endpgfgraphicnamed
=%
	\beginpgfgraphicnamed{TikZit/k2srt}
	\InputIfFileExists{TikZit/k2srt.tikz}{}{\input{./figures/TikZit/k2srt.tikz}}%
	\endpgfgraphicnamed
 (1i) 
   \end{proposition} 
 \begin{proof}
	\beginpgfgraphicnamed{TikZit/k2prf}
	\InputIfFileExists{TikZit/k2prf.tikz}{}{\input{./figures/TikZit/k2prf.tikz}}%
	\endpgfgraphicnamed
  
 (1i) follows when setting $a=e^{i\alpha}$.
 \end{proof}

  \begin{lemma}
	\beginpgfgraphicnamed{TikZit/scalartimes}
	\InputIfFileExists{TikZit/scalartimes.tikz}{}{\input{./figures/TikZit/scalartimes.tikz}}%
	\endpgfgraphicnamed
 (Sca) 
      \end{lemma}    
   \begin{proof}
	\beginpgfgraphicnamed{TikZit/scalartimesprf}
	\InputIfFileExists{TikZit/scalartimesprf.tikz}{}{\input{./figures/TikZit/scalartimesprf.tikz}}%
	\endpgfgraphicnamed
  
   \end{proof}

 \begin{corollary}\label{zeroiscalarempty}
	\beginpgfgraphicnamed{TikZit/zeroscalarempty}
	\InputIfFileExists{TikZit/zeroscalarempty.tikz}{}{\input{./figures/TikZit/zeroscalarempty.tikz}}%
	\endpgfgraphicnamed
  (Zos)
   \end{corollary} 

 \begin{proof}
 Let $a=1$ in (Sca) and use (Ivs) and Lemma \ref{redpigdot}.    
     \end{proof}

  \begin{lemma}
	\beginpgfgraphicnamed{TikZit/scalartimesgeneral}
	\begin{tikzpicture}
	\begin{pgfonlayer}{nodelayer}
		\node [style=gbox] (0) at (-1, 0) {$a$};
		\node [style=gbox] (1) at (-0.5, 0) {$b$};
		\node [style=gbox] (2) at (1, 0) {${\scriptstyle (a+1)(b+1)-1}$};
		\node [style=none] (3) at (0, 0) {$=$};
	\end{pgfonlayer}
\end{tikzpicture}}%
	\endpgfgraphicnamed
 (Sml)  
      \end{lemma}   
      
 \begin{proof}
	\beginpgfgraphicnamed{TikZit/scalartimesgeneralprf}
	\InputIfFileExists{TikZit/scalartimesgeneralprf.tikz}{}{\input{./figures/TikZit/scalartimesgeneralprf.tikz}}%
	\endpgfgraphicnamed
  
   \end{proof}   

  \begin{lemma}
	\beginpgfgraphicnamed{TikZit/1ovsrt2}
	\InputIfFileExists{TikZit/1ovsrt2.tikz}{}{\input{./figures/TikZit/1ovsrt2.tikz}}%
	\endpgfgraphicnamed
 (Irt) 
       \end{lemma}   
 \begin{proof}
	\beginpgfgraphicnamed{TikZit/1ovsrt2prf}
	\InputIfFileExists{TikZit/1ovsrt2prf.tikz}{}{\input{./figures/TikZit/1ovsrt2prf.tikz}}%
	\endpgfgraphicnamed
  
    \end{proof}

   \begin{proposition}\label{1jprf}
	\beginpgfgraphicnamed{TikZit/newemptyrl}
	\InputIfFileExists{TikZit/newemptyrl.tikz}{}{\input{./figures/TikZit/newemptyrl.tikz}}%
	\endpgfgraphicnamed
 (1j) 
   \end{proposition}   
      \begin{proof}
	\beginpgfgraphicnamed{TikZit/1jprf}
	\InputIfFileExists{TikZit/1jprf.tikz}{}{\input{./figures/TikZit/1jprf.tikz}}%
	\endpgfgraphicnamed
  
 \end{proof}
     
 Note that     %
	\beginpgfgraphicnamed{TikZit/lambdabxequalgbx}
	\InputIfFileExists{TikZit/lambdabxequalgbx.tikz}{}{\input{./figures/TikZit/lambdabxequalgbx.tikz}}%
	\endpgfgraphicnamed
.  Then (2a) and (2c) follow directly from (S1). (2b) directly follows  from (S2).

   \begin{proposition}\label{2eprf}
	\beginpgfgraphicnamed{TikZit/zx2e}
	\InputIfFileExists{TikZit/zx2e.tikz}{}{\input{./figures/TikZit/zx2e.tikz}}%
	\endpgfgraphicnamed
 (2e) 
    \end{proposition}   
  \begin{proof}
  %$$ \tikzfig{TikZit/tr8primeprf2}  $$
$$ %
	\beginpgfgraphicnamed{TikZit/tr8primeprf3}
	\InputIfFileExists{TikZit/tr8primeprf3.tikz}{}{\input{./figures/TikZit/tr8primeprf3.tikz}}%
	\endpgfgraphicnamed
  $$
    \end{proof}   
     
(2f) and (2g) are exactly (Bas0) and (Bas1) respectively. 

  \begin{lemma}
	\beginpgfgraphicnamed{TikZit/redpitogreen}
	\InputIfFileExists{TikZit/redpitogreen.tikz}{}{\input{./figures/TikZit/redpitogreen.tikz}}%
	\endpgfgraphicnamed
 (Bas1')
  \end{lemma}   
This can be directly obtained by plugging a triangle inverse on both sides of (Bas1).

   \begin{proposition}\label{2iprf}
	\beginpgfgraphicnamed{TikZit/tr5prime}
	\InputIfFileExists{TikZit/tr5prime.tikz}{}{\input{./figures/TikZit/tr5prime.tikz}}%
	\endpgfgraphicnamed
 (2i) 
    \end{proposition}  
  \begin{proof}
$$ %
	\beginpgfgraphicnamed{TikZit/tr5primeprf}
	\InputIfFileExists{TikZit/tr5primeprf.tikz}{}{\input{./figures/TikZit/tr5primeprf.tikz}}%
	\endpgfgraphicnamed
  $$
    \end{proof}

  \begin{proposition}\label{2jprf}
	\beginpgfgraphicnamed{TikZit/trianglehopf}
	\InputIfFileExists{TikZit/trianglehopf.tikz}{}{\input{./figures/TikZit/trianglehopf.tikz}}%
	\endpgfgraphicnamed
 (2j) 
   \end{proposition}  
   
   \begin{proof}
$$ %
	\beginpgfgraphicnamed{TikZit/Hopftrprf}
	\InputIfFileExists{TikZit/Hopftrprf.tikz}{}{\input{./figures/TikZit/Hopftrprf.tikz}}%
	\endpgfgraphicnamed
  $$
        \end{proof} 
 
  \begin{corollary}
 \begin{equation}\label{trianglehopflip}
	\beginpgfgraphicnamed{TikZit/trianglehopfflip}
	\InputIfFileExists{TikZit/trianglehopfflip.tikz}{}{\input{./figures/TikZit/trianglehopfflip.tikz}}%
	\endpgfgraphicnamed
 
   \end{equation} 
    \end{corollary}
  \begin{proof}
 $$%
	\beginpgfgraphicnamed{TikZit/trianglehopfflipprf}
	\InputIfFileExists{TikZit/trianglehopfflipprf.tikz}{}{\input{./figures/TikZit/trianglehopfflipprf.tikz}}%
	\endpgfgraphicnamed
  $$
   \end{proof}   

 \begin{proposition}\label{2kprf}
	\beginpgfgraphicnamed{TikZit/2triangleup}
	\InputIfFileExists{TikZit/2triangleup.tikz}{}{\input{./figures/TikZit/2triangleup.tikz}}%
	\endpgfgraphicnamed
 (2k) 
   \end{proposition}  
   
 \begin{proof}
$$  %
	\beginpgfgraphicnamed{TikZit/2triangleupprf}
	\InputIfFileExists{TikZit/2triangleupprf.tikz}{}{\input{./figures/TikZit/2triangleupprf.tikz}}%
	\endpgfgraphicnamed
$$
  \end{proof}

 \begin{proposition}\label{2mprf}
	\beginpgfgraphicnamed{TikZit/2triangledeloopnopi}
	\InputIfFileExists{TikZit/2triangledeloopnopi.tikz}{}{\input{./figures/TikZit/2triangledeloopnopi.tikz}}%
	\endpgfgraphicnamed
 (2m)
 \end{proposition}  
 
  \begin{proof}
 $$%
	\beginpgfgraphicnamed{TikZit/brk1prf}
	\InputIfFileExists{TikZit/brk1prf.tikz}{}{\input{./figures/TikZit/brk1prf.tikz}}%
	\endpgfgraphicnamed
  $$
   \end{proof}   
 
  \begin{corollary}
  \begin{equation}
	\beginpgfgraphicnamed{TikZit/2triangledeloopnopiflip}
	\InputIfFileExists{TikZit/2triangledeloopnopiflip.tikz}{}{\input{./figures/TikZit/2triangledeloopnopiflip.tikz}}%
	\endpgfgraphicnamed
 \quad (Brk1') 
   \end{equation}  
    \end{corollary}  
     \begin{proof}
 $$%
	\beginpgfgraphicnamed{TikZit/brk1primeprf}
	\InputIfFileExists{TikZit/brk1primeprf.tikz}{}{\input{./figures/TikZit/brk1primeprf.tikz}}%
	\endpgfgraphicnamed
  $$
 The other part can be proved by symmetry. 
   \end{proof}   

(2n) is just the rule (Pcy).

 \begin{lemma}
	\beginpgfgraphicnamed{TikZit/trianglegdpicopy}
	\InputIfFileExists{TikZit/trianglegdpicopy.tikz}{}{\input{./figures/TikZit/trianglegdpicopy.tikz}}%
	\endpgfgraphicnamed
 (Zrp)
    \end{lemma}
  By the rule (Suc), it is clear that (Zrp) is equivalent to the rule (Zero).
 
 \begin{lemma}
	\beginpgfgraphicnamed{TikZit/zerodecom}
	\InputIfFileExists{TikZit/zerodecom.tikz}{}{\input{./figures/TikZit/zerodecom.tikz}}%
	\endpgfgraphicnamed
 (Zero') 
  \end{lemma}
 \begin{proof}
% \tikzfig{TikZit/zerodecomprf}  
$$ %
	\beginpgfgraphicnamed{TikZit/zerodecomprf2}
	\InputIfFileExists{TikZit/zerodecomprf2.tikz}{}{\input{./figures/TikZit/zerodecomprf2.tikz}}%
	\endpgfgraphicnamed
  $$
     \end{proof}

 \begin{lemma}
	\beginpgfgraphicnamed{TikZit/tr3equivl}
	\InputIfFileExists{TikZit/tr3equivl.tikz}{}{\input{./figures/TikZit/tr3equivl.tikz}}%
	\endpgfgraphicnamed
 (Bas0')
   \end{lemma}
 \begin{proof}
$$ %
	\beginpgfgraphicnamed{TikZit/tr2primeprf}
	\InputIfFileExists{TikZit/tr2primeprf.tikz}{}{\input{./figures/TikZit/tr2primeprf.tikz}}%
	\endpgfgraphicnamed
  $$
    \end{proof}

 \begin{lemma}
  \begin{equation}\label{TR4g}
	\beginpgfgraphicnamed{TikZit/tr4g}
	\InputIfFileExists{TikZit/tr4g.tikz}{}{\input{./figures/TikZit/tr4g.tikz}}%
	\endpgfgraphicnamed
 
   \end{equation}
    \end{lemma}
 \begin{proof}
 $$%
	\beginpgfgraphicnamed{TikZit/tr4gprf}
	\InputIfFileExists{TikZit/tr4gprf.tikz}{}{\input{./figures/TikZit/tr4gprf.tikz}}%
	\endpgfgraphicnamed
$$  
  \end{proof}  
  %\tikzfig{TikZit/lemma02v} 

  \begin{lemma}
  \begin{equation}\label{Hopfgtr}
	\beginpgfgraphicnamed{TikZit/trianglehopfgreen}
	\InputIfFileExists{TikZit/trianglehopfgreen.tikz}{}{\input{./figures/TikZit/trianglehopfgreen.tikz}}%
	\endpgfgraphicnamed

  \end{equation}
    \end{lemma}
 \begin{proof}
$$ %
	\beginpgfgraphicnamed{TikZit/trianglehopfgreenprf}
	\InputIfFileExists{TikZit/trianglehopfgreenprf.tikz}{}{\input{./figures/TikZit/trianglehopfgreenprf.tikz}}%
	\endpgfgraphicnamed
  $$
  \end{proof}

   \begin{lemma}
  \begin{equation}\label{TR19}
	\beginpgfgraphicnamed{TikZit/trianglecopylr}
	\InputIfFileExists{TikZit/trianglecopylr.tikz}{}{\input{./figures/TikZit/trianglecopylr.tikz}}%
	\endpgfgraphicnamed

  \end{equation}
    \end{lemma}
 \begin{proof}
$$ %
	\beginpgfgraphicnamed{TikZit/trianglecopylrprf}
	\InputIfFileExists{TikZit/trianglecopylrprf.tikz}{}{\input{./figures/TikZit/trianglecopylrprf.tikz}}%
	\endpgfgraphicnamed
  $$
   \end{proof}

  \begin{lemma}
   \begin{equation*}
	\beginpgfgraphicnamed{TikZit/equivalentaddrule}
	\InputIfFileExists{TikZit/equivalentaddrule.tikz}{}{\input{./figures/TikZit/equivalentaddrule.tikz}}%
	\endpgfgraphicnamed
 \quad (AD') 
   \end{equation*}    
     \end{lemma}
    \begin{proof}
    If  $b\neq 0$, then 
$$%
	\beginpgfgraphicnamed{TikZit/equivalentaddruleprf}
	\InputIfFileExists{TikZit/equivalentaddruleprf.tikz}{}{\input{./figures/TikZit/equivalentaddruleprf.tikz}}%
	\endpgfgraphicnamed
  $$
 If $b= 0$, then
$$%
	\beginpgfgraphicnamed{TikZit/equivalentaddruleprf0}
	\InputIfFileExists{TikZit/equivalentaddruleprf0.tikz}{}{\input{./figures/TikZit/equivalentaddruleprf0.tikz}}%
	\endpgfgraphicnamed
  $$
   \end{proof}   
 
 \begin{proposition}
	\beginpgfgraphicnamed{TikZit/additiongbx}
	\InputIfFileExists{TikZit/additiongbx.tikz}{}{\input{./figures/TikZit/additiongbx.tikz}}%
	\endpgfgraphicnamed
 (2o) 
  \end{proposition}
  \begin{proof}
$$ %
	\beginpgfgraphicnamed{TikZit/addprf}
	\InputIfFileExists{TikZit/addprf.tikz}{}{\input{./figures/TikZit/addprf.tikz}}%
	\endpgfgraphicnamed
  $$
   \end{proof}

 \begin{proposition}
	\beginpgfgraphicnamed{TikZit/tr10prime}
	\InputIfFileExists{TikZit/tr10prime.tikz}{}{\input{./figures/TikZit/tr10prime.tikz}}%
	\endpgfgraphicnamed
 (2h) 
   \end{proposition}
%  \tikzfig{TikZit/definitionTriangleInverse2} (IVT) 
  \begin{proof}
$$ %
	\beginpgfgraphicnamed{TikZit/invtriprf}
	\InputIfFileExists{TikZit/invtriprf.tikz}{}{\input{./figures/TikZit/invtriprf.tikz}}%
	\endpgfgraphicnamed
  $$
(2h) follows immediately. 
%\tikzfig{TikZit/TR10primederive22new} 
    \end{proof}   
Clearly, (2h) is equivalent to  %
	\beginpgfgraphicnamed{TikZit/definitionTriangleInverse2}
	\InputIfFileExists{TikZit/definitionTriangleInverse2.tikz}{}{\input{./figures/TikZit/definitionTriangleInverse2.tikz}}%
	\endpgfgraphicnamed
 (IVT).      
     
 \begin{corollary}
   \begin{equation}\label{pitinvcomut}
	\beginpgfgraphicnamed{TikZit/rpitrinverse}
	\InputIfFileExists{TikZit/rpitrinverse.tikz}{}{\input{./figures/TikZit/rpitrinverse.tikz}}%
	\endpgfgraphicnamed
 
     \end{equation}
 \end{corollary}
  \begin{proof}
$$ %
	\beginpgfgraphicnamed{TikZit/rpitrinverseprf}
	\InputIfFileExists{TikZit/rpitrinverseprf.tikz}{}{\input{./figures/TikZit/rpitrinverseprf.tikz}}%
	\endpgfgraphicnamed
  $$
    \end{proof}

 \begin{proposition}
	\beginpgfgraphicnamed{TikZit/gpiintriangles}
	\InputIfFileExists{TikZit/gpiintriangles.tikz}{}{\input{./figures/TikZit/gpiintriangles.tikz}}%
	\endpgfgraphicnamed
 (2d)  
     \end{proposition}
  \begin{proof}
$$ %
	\beginpgfgraphicnamed{TikZit/gpiintrianglesprf}
	\InputIfFileExists{TikZit/gpiintrianglesprf.tikz}{}{\input{./figures/TikZit/gpiintrianglesprf.tikz}}%
	\endpgfgraphicnamed
  $$
    \end{proof}  

   \begin{lemma}
  \begin{equation}\label{TRPh}
	\beginpgfgraphicnamed{TikZit/trianglephase}
	\InputIfFileExists{TikZit/trianglephase.tikz}{}{\input{./figures/TikZit/trianglephase.tikz}}%
	\endpgfgraphicnamed
 
   \end{equation}
    \end{lemma}
 \begin{proof}
$$ %
	\beginpgfgraphicnamed{TikZit/trianglephaseprf}
	\InputIfFileExists{TikZit/trianglephaseprf.tikz}{}{\input{./figures/TikZit/trianglephaseprf.tikz}}%
	\endpgfgraphicnamed
  $$
       \end{proof}

  \begin{lemma}
	\beginpgfgraphicnamed{TikZit/h2triandecom2}
	\InputIfFileExists{TikZit/h2triandecom2.tikz}{}{\input{./figures/TikZit/h2triandecom2.tikz}}%
	\endpgfgraphicnamed
 (H2)  
   \end{lemma}
 \begin{proof}
$$ %
	\beginpgfgraphicnamed{TikZit/h2triandecom2prf}
	\InputIfFileExists{TikZit/h2triandecom2prf.tikz}{}{\input{./figures/TikZit/h2triandecom2prf.tikz}}%
	\endpgfgraphicnamed
$$  
         \end{proof}

   \begin{lemma}
  \begin{equation}\label{TRH}
	\beginpgfgraphicnamed{TikZit/trihbx2}
	\InputIfFileExists{TikZit/trihbx2.tikz}{}{\input{./figures/TikZit/trihbx2.tikz}}%
	\endpgfgraphicnamed
 
   \end{equation}
    \end{lemma}
  \begin{proof}
 $$%
	\beginpgfgraphicnamed{TikZit/trihbx2prf}
	\InputIfFileExists{TikZit/trihbx2prf.tikz}{}{\input{./figures/TikZit/trihbx2prf.tikz}}%
	\endpgfgraphicnamed
  $$
        \end{proof}

    \begin{proposition}
	\beginpgfgraphicnamed{TikZit/2triangledown}
	\InputIfFileExists{TikZit/2triangledown.tikz}{}{\input{./figures/TikZit/2triangledown.tikz}}%
	\endpgfgraphicnamed
 (2l) 
   \end{proposition}
  \begin{proof}
 $$%
	\beginpgfgraphicnamed{TikZit/2triangledownproof2}
	\InputIfFileExists{TikZit/2triangledownproof2.tikz}{}{\input{./figures/TikZit/2triangledownproof2.tikz}}%
	\endpgfgraphicnamed
  $$
      \end{proof}   

\section*{Proof of Theorem \ref{zhcomplete}}

 For simplicity, we use the following notation.
  $$  %
	\beginpgfgraphicnamed{TikZit/andshortnote}
	\InputIfFileExists{TikZit/andshortnote.tikz}{}{\input{./figures/TikZit/andshortnote.tikz}}%
	\endpgfgraphicnamed
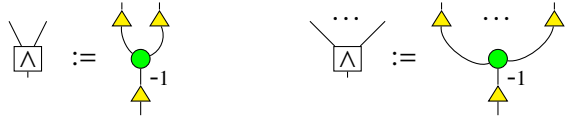
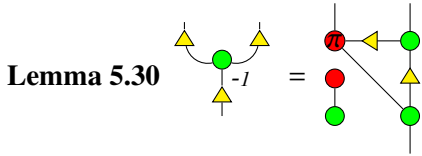
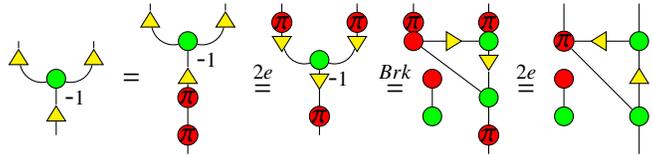
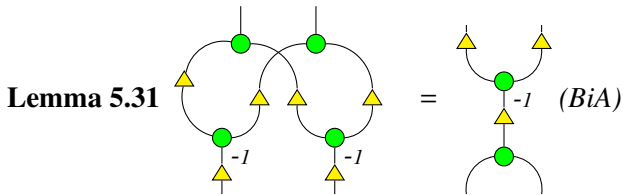
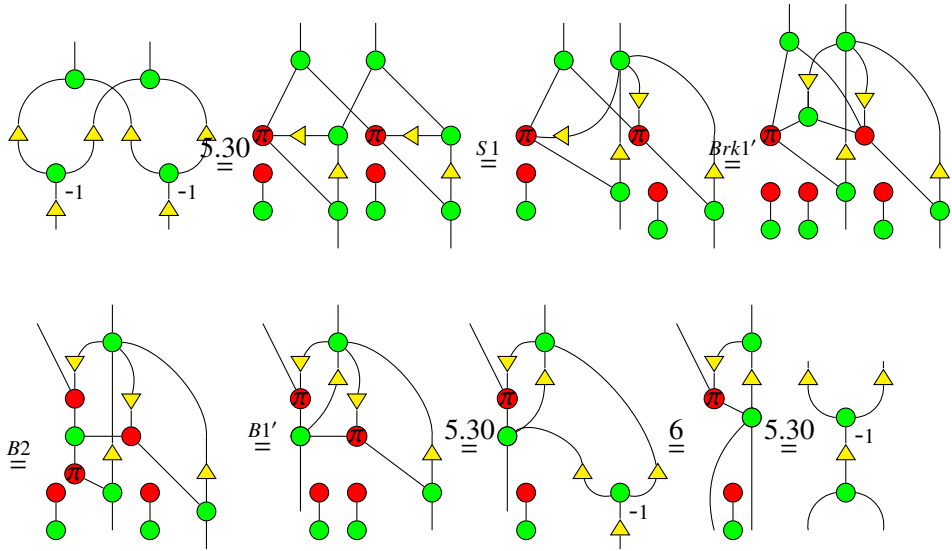
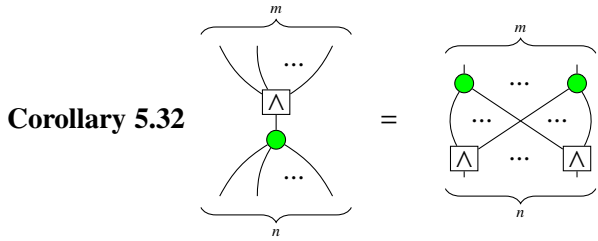
 $$

   \begin{lemma}\label{andgate2v}
	\beginpgfgraphicnamed{TikZit/andgate2vs}
	\InputIfFileExists{TikZit/andgate2vs.tikz}{}{\input{./figures/TikZit/andgate2vs.tikz}}%
	\endpgfgraphicnamed
 
 \end{lemma}
   \begin{proof}
   $$  %
	\beginpgfgraphicnamed{TikZit/andgate2vsprf}
	\InputIfFileExists{TikZit/andgate2vsprf.tikz}{}{\input{./figures/TikZit/andgate2vsprf.tikz}}%
	\endpgfgraphicnamed
 $$
  \end{proof}

   \begin{lemma}\label{andbial}
	\beginpgfgraphicnamed{TikZit/ruleA3}
	\InputIfFileExists{TikZit/ruleA3.tikz}{}{\input{./figures/TikZit/ruleA3.tikz}}%
	\endpgfgraphicnamed
 (BiA)
 \end{lemma}
  
   \begin{proof}
 $$  %
	\beginpgfgraphicnamed{TikZit/andbialgebraprf}
	\InputIfFileExists{TikZit/andbialgebraprf.tikz}{}{\input{./figures/TikZit/andbialgebraprf.tikz}}%
	\endpgfgraphicnamed
 $$
  \end{proof}
  
    \begin{corollary}\label{generalbialgebra}
	\beginpgfgraphicnamed{TikZit/generalBiA}
	\InputIfFileExists{TikZit/generalBiA.tikz}{}{\input{./figures/TikZit/generalBiA.tikz}}%
	\endpgfgraphicnamed
 
 \end{corollary}
  This can be proved by induction, see (L1) in \cite{bobanthonywang}.

    Note that 
   \begin{equation}\label{andzhzx}
   \left\llbracket~
	\beginpgfgraphicnamed{TikZit/andzh}
	\InputIfFileExists{TikZit/andzh.tikz}{}{\input{./figures/TikZit/andzh.tikz}}%
	\endpgfgraphicnamed
 
	~\right\rrbracket_{HX}= %
	\beginpgfgraphicnamed{TikZit/andzx}
	\InputIfFileExists{TikZit/andzx.tikz}{}{\input{./figures/TikZit/andzx.tikz}}%
	\endpgfgraphicnamed
 
 \end{equation}
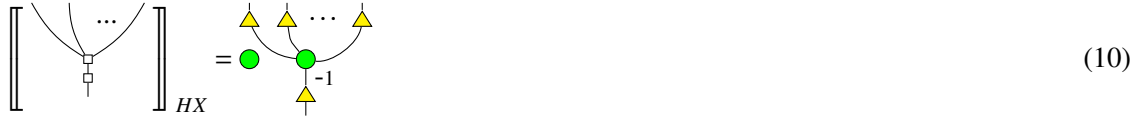
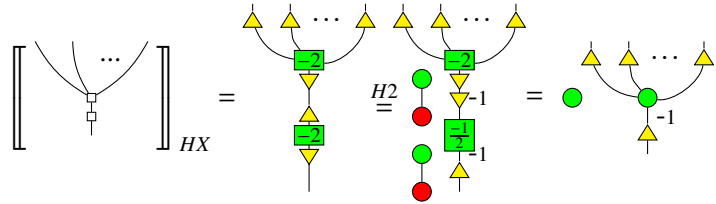
    \begin{proof}
   $$  \left\llbracket~
	\beginpgfgraphicnamed{TikZit/andzh}
	\InputIfFileExists{TikZit/andzh.tikz}{}{\input{./figures/TikZit/andzh.tikz}}%
	\endpgfgraphicnamed
 
	~\right\rrbracket_{HX}= %
	\beginpgfgraphicnamed{TikZit/andzhzxprf}
	\InputIfFileExists{TikZit/andzhzxprf.tikz}{}{\input{./figures/TikZit/andzhzxprf.tikz}}%
	\endpgfgraphicnamed
 $$
  \end{proof}

 Then  it is clear that the ZH rule (BA2) follows directly from equation \ref{andzhzx} and 
  Corollary \ref{generalbialgebra}.

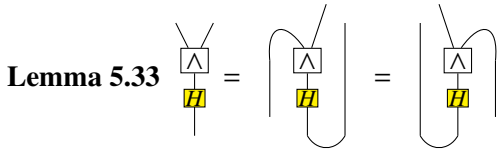
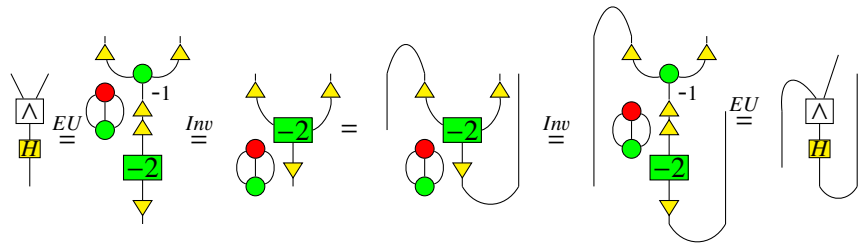
   \begin{lemma}\label{andgatehadam}
	\beginpgfgraphicnamed{TikZit/andgatehadama}
	\InputIfFileExists{TikZit/andgatehadama.tikz}{}{\input{./figures/TikZit/andgatehadama.tikz}}%
	\endpgfgraphicnamed
 
 \end{lemma}
  \begin{proof}
 $$  %
	\beginpgfgraphicnamed{TikZit/andgatehadamaprf}
	\InputIfFileExists{TikZit/andgatehadamaprf.tikz}{}{\input{./figures/TikZit/andgatehadamaprf.tikz}}%
	\endpgfgraphicnamed
 $$
 The other equality can be obtained by symmetry. 
  \end{proof}
  
 \begin{lemma}\label{distribute}
	\beginpgfgraphicnamed{TikZit/ruleA1_Draft}
	\InputIfFileExists{TikZit/ruleA1_Draft.tikz}{}{\input{./figures/TikZit/ruleA1_Draft.tikz}}%
	\endpgfgraphicnamed
 (Dis)
 \end{lemma}
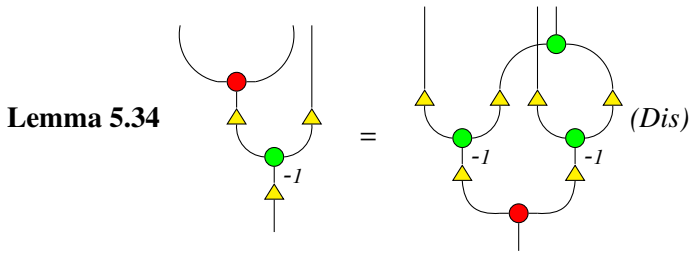
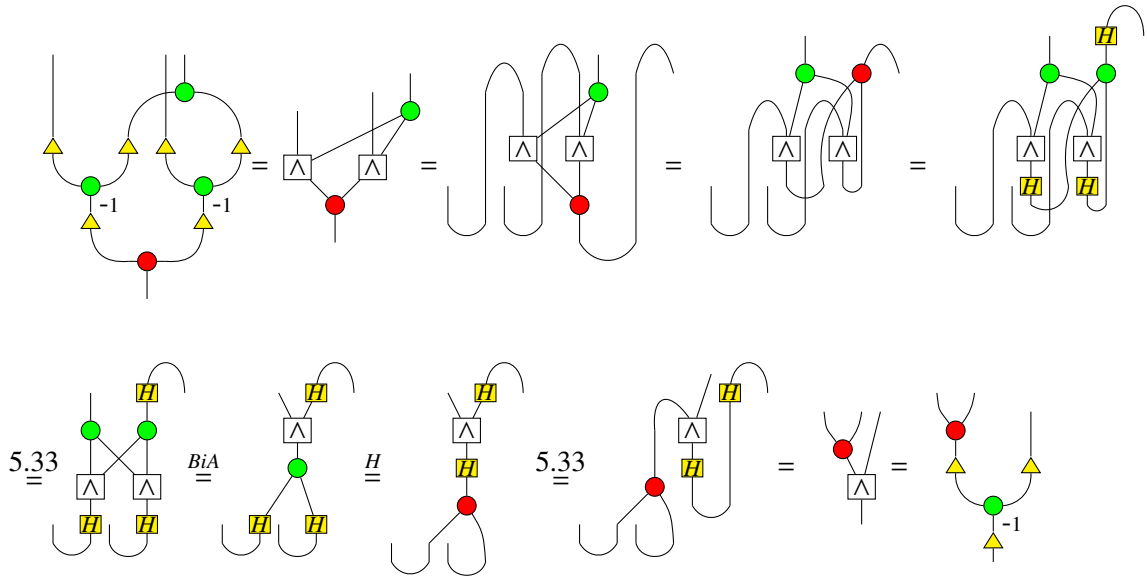
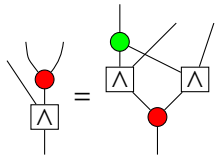
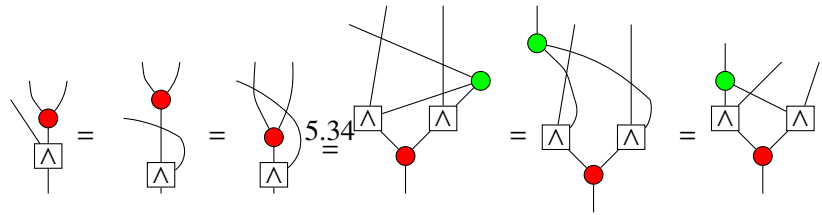
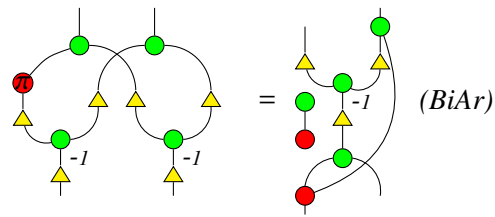
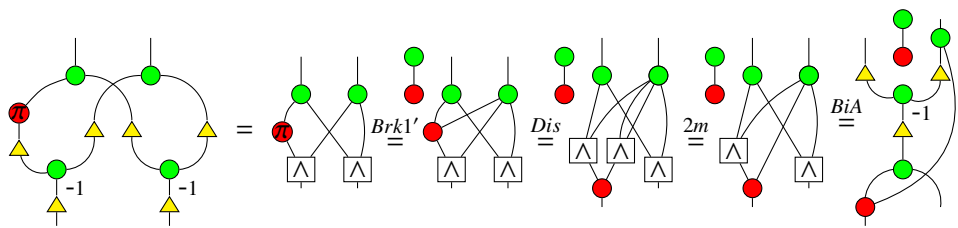
 \begin{proof}
 $$  %
	\beginpgfgraphicnamed{TikZit/distributionprf}
	\InputIfFileExists{TikZit/distributionprf.tikz}{}{\input{./figures/TikZit/distributionprf.tikz}}%
	\endpgfgraphicnamed
 $$
  \end{proof}
  
   \begin{corollary}\label{distribute2}
	\beginpgfgraphicnamed{TikZit/distribute2}
	\InputIfFileExists{TikZit/distribute2.tikz}{}{\input{./figures/TikZit/distribute2.tikz}}%
	\endpgfgraphicnamed
 
 \end{corollary}
  \begin{proof}
 $$  %
	\beginpgfgraphicnamed{TikZit/distribute2prf}
	\InputIfFileExists{TikZit/distribute2prf.tikz}{}{\input{./figures/TikZit/distribute2prf.tikz}}%
	\endpgfgraphicnamed
 $$
  \end{proof}

    \begin{lemma}
	\beginpgfgraphicnamed{TikZit/andbialgebraredpi}
	\InputIfFileExists{TikZit/andbialgebraredpi.tikz}{}{\input{./figures/TikZit/andbialgebraredpi.tikz}}%
	\endpgfgraphicnamed
 (BiAr)
   \end{lemma} 
    \begin{proof}
   \[  %
	\beginpgfgraphicnamed{TikZit/andbialgebraredpiprf}
	\InputIfFileExists{TikZit/andbialgebraredpiprf.tikz}{}{\input{./figures/TikZit/andbialgebraredpiprf.tikz}}%
	\endpgfgraphicnamed
 \]
      \end{proof}  
      
     \begin{proposition} (Derivation of  (HS1)) \label{zhs1}
 $$ ZX\vdash
	\left\llbracket~
	\beginpgfgraphicnamed{TikZit/zhhs1}
	\InputIfFileExists{TikZit/zhhs1.tikz}{}{\input{./figures/TikZit/zhhs1.tikz}}%
	\endpgfgraphicnamed
 
	~\right\rrbracket_{HX}$$
 \end{proposition}         
  \begin{proof}
   \[  %
	\beginpgfgraphicnamed{TikZit/zhhs1prf}
	\InputIfFileExists{TikZit/zhhs1prf.tikz}{}{\input{./figures/TikZit/zhhs1prf.tikz}}%
	\endpgfgraphicnamed
 \]
      \end{proof}

      \begin{lemma}
	\beginpgfgraphicnamed{TikZit/anddflipwitha2}
	\InputIfFileExists{TikZit/anddflipwitha2.tikz}{}{\input{./figures/TikZit/anddflipwitha2.tikz}}%
	\endpgfgraphicnamed
 (Brkp)      
   \end{lemma} 
   \begin{proof}
   \[  %
	\beginpgfgraphicnamed{TikZit/anddflipwitha2prf}
	\InputIfFileExists{TikZit/anddflipwitha2prf.tikz}{}{\input{./figures/TikZit/anddflipwitha2prf.tikz}}%
	\endpgfgraphicnamed
 \]
   Therefore,
      \[  %
	\beginpgfgraphicnamed{TikZit/anddflipwitha2prf2}
	\InputIfFileExists{TikZit/anddflipwitha2prf2.tikz}{}{\input{./figures/TikZit/anddflipwitha2prf2.tikz}}%
	\endpgfgraphicnamed
 \]
     \end{proof}   
     
    \begin{proposition}(Derivation of  (A))\label{zha}
   $$ ZX\vdash
	\left\llbracket~
	\beginpgfgraphicnamed{TikZit/zha}
	\InputIfFileExists{TikZit/zha.tikz}{}{\input{./figures/TikZit/zha.tikz}}%
	\endpgfgraphicnamed
 
	~\right\rrbracket_{HX}$$
%\tikzfig{TikZit/zha}  \quad (A)
 \end{proposition}    
      \begin{proof}
      First we have 
        \[  %
	\beginpgfgraphicnamed{TikZit/zhaprfb}
	\InputIfFileExists{TikZit/zhaprfb.tikz}{}{\input{./figures/TikZit/zhaprfb.tikz}}%
	\endpgfgraphicnamed
 \]
      If $b=0$, then
        \[  %
	\beginpgfgraphicnamed{TikZit/zhaprf0}
	\InputIfFileExists{TikZit/zhaprf0.tikz}{}{\input{./figures/TikZit/zhaprf0.tikz}}%
	\endpgfgraphicnamed
 \]
      
     If $b\neq 0$, then we have
   \[  %
	\beginpgfgraphicnamed{TikZit/zhaprf1}
	\InputIfFileExists{TikZit/zhaprf1.tikz}{}{\input{./figures/TikZit/zhaprf1.tikz}}%
	\endpgfgraphicnamed
 \]
      \end{proof}  
     
       \begin{proposition}(Derivation of  (I))\label{zhi}
   $$ ZX\vdash
	\left\llbracket~
	\beginpgfgraphicnamed{TikZit/zhi}
	\InputIfFileExists{TikZit/zhi.tikz}{}{\input{./figures/TikZit/zhi.tikz}}%
	\endpgfgraphicnamed
 
	~\right\rrbracket_{HX}$$
   \end{proposition}    
         \begin{proof}
   \[  %
	\beginpgfgraphicnamed{TikZit/zhiprf}
	\InputIfFileExists{TikZit/zhiprf.tikz}{}{\input{./figures/TikZit/zhiprf.tikz}}%
	\endpgfgraphicnamed
 \]
      \end{proof}  
     
       \begin{proposition}(Derivation of  (O))\label{zho}
   $$ ZX\vdash
	\left\llbracket~
	\beginpgfgraphicnamed{TikZit/zho}
	\InputIfFileExists{TikZit/zho.tikz}{}{\input{./figures/TikZit/zho.tikz}}%
	\endpgfgraphicnamed
 
	~\right\rrbracket_{HX}$$
   \end{proposition}  
       \begin{proof}
   \[  %
	\beginpgfgraphicnamed{TikZit/zhoprf}
	\InputIfFileExists{TikZit/zhoprf.tikz}{}{\input{./figures/TikZit/zhoprf.tikz}}%
	\endpgfgraphicnamed
 \]
      \end{proof}  
     
%\fi

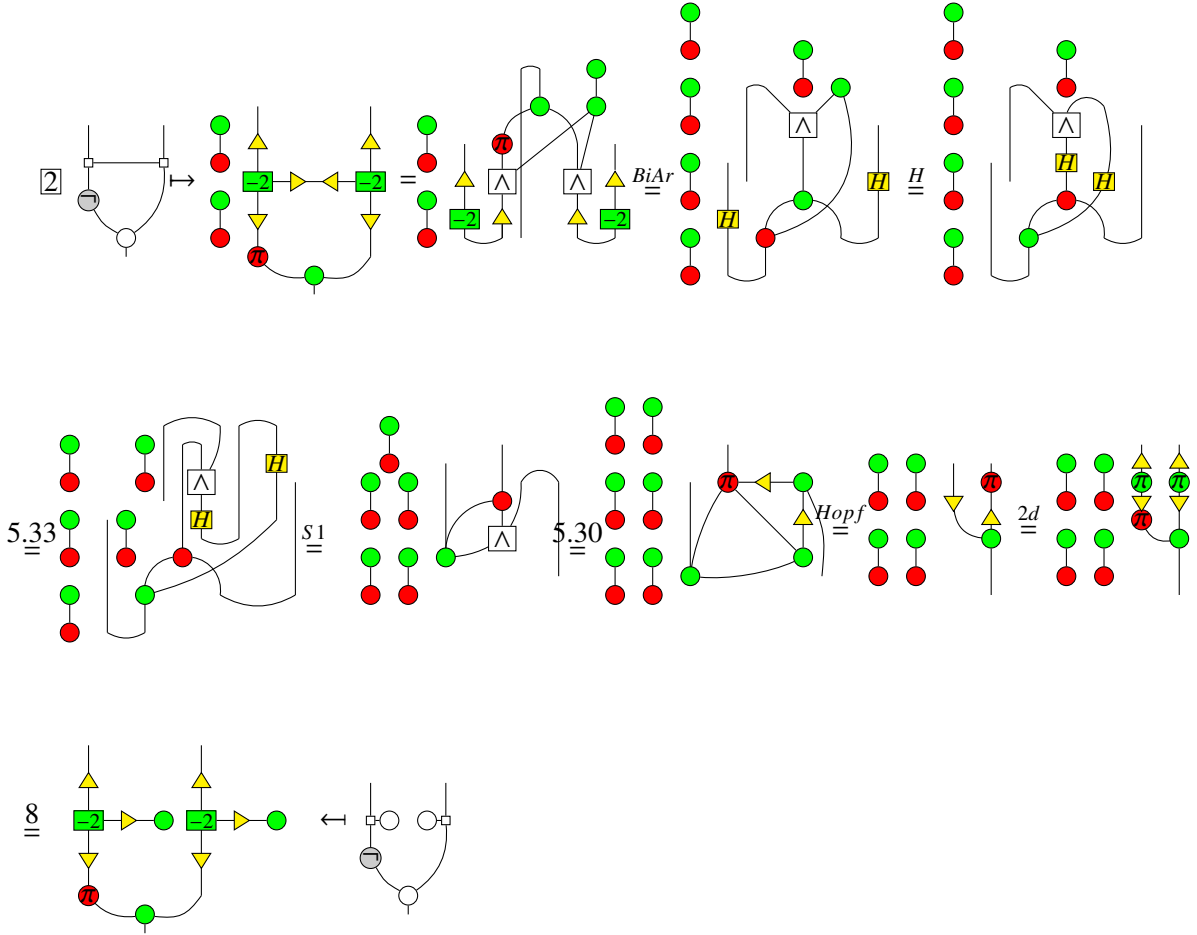
\end{document}